\documentclass{llncs}

\usepackage{todonotes}
\usepackage{amsmath}
\usepackage{amssymb}
\usepackage{tikz}
\usepackage{wrapfig}

\usepackage[pdftex]{hyperref}

\usepackage[amsfonts]{logic-thm}

\newcommand{\Sync}{\mathit{Sync}}
\newcommand{\CG}{\Cc\Gg}

\newcommand{\rlp}{r_{l+1}}
\newcommand{\brlp}{\bar{\rlp}}
\newcommand{\cc}{\mathfrak{c}}

\newcommand{\sskip}{\mathit{skip}}
\newcommand{\inctest}{\searrow}
\newcommand{\eqtest}{\downarrow}
\newcommand{\Tower}{\mathit{Tower}}
\newcommand{\state}{\mathit{state}}

\newcommand{\loc}{\mathit{dom}}
\newcommand{\Plays}{\mathit{Plays}}
\newcommand{\view}{\mathit{view}}

\newcommand{\dar}{\downarrow}

\newcommand{\Ssys}{\S^{sys}}
\newcommand{\Senv}{\S^{env}}

\newcommand{\edge}{-\!\!\!-\!\!\!-\!\!\!-\!\!\!-}
\newcommand{\ch}{\mathop{ch}}

\newcommand{\igw}[1]{}%{\todo{#1}}
\newcommand{\anca}[1]{}%{\todo[color=green]{#1}}
\newcommand{\hugo}[1]{}%{\todo[color=yellow]{#1}}

\renewenvironment{proof}{{\em Proof. }}{\nopagebreak
  \hspace*{\fill}$\Box$}

\begin{document}

\title{Asynchronous Games over Tree Architectures}
%Controling asynchronous automata: the acyclic case}

\author{Blaise Genest$^{1}$, Hugo Gimbert$^2$, Anca Muscholl$^{2}$,
  Igor Walukiewicz$^2$}

 \institute{$^{1}$ IRISA, CNRS, Rennes, France\\
% $^{2}$ CNRS, IRISA UMR, joint with Université Rennes I, France\\
 $^{2}$ LaBRI, CNRS/Université Bordeaux, France \\
   }

\maketitle

\begin{abstract}
 We consider the distributed control problem in the setting of
  Zielonka asynchronous automata. 
  Such automata are compositions of finite processes
  communicating via shared actions and evolving asynchronously. Most
  importantly, processes participating in a shared action can exchange
  complete information about their causal past. This gives more power
  to controllers, and avoids simple pathological undecidable cases as
  in the setting of Pnueli
  and Rosner.
  We show the decidability of the control problem for
  Zielonka automata over acyclic communication architectures.  We
  provide also a matching lower bound, which is $l$-fold exponential,
  $l$ being the height of the architecture tree.
\end{abstract}

\section{Introduction}\label{sec:intro}
Synthesis is by now well understood in the case of sequential systems.
It is useful for constructing small, yet safe, critical modules.
Initially, the {\em synthesis problem} was stated by Church, who asked
for an algorithm to construct devices transforming sequences of input
bits into sequences of output bits in a way required by a
specification~\cite{church62}. Later Ramadge and Wonham proposed the
{\em supervisory control} formulation, where a plant and a
specification are given, and a controller should be designed such that
its product with the plant satisfies the specification~\cite{RW89}. So
control means restricting the behavior of the plant. Synthesis is the
particular case of control where the plant allows for every possible
behavior.

For synthesis of {\em distributed} systems, a common belief is that the
problem is in general undecidable, referring to work by Pnueli and
Rosner~\cite{PR90}. They extended Church's 
formulation to an architecture of {\em synchronously} communicating
processes, that exchange messages through one slot communication
channels. Undecidability in this setting comes mainly from~\emph{partial
information}: specifications permit to control the flow of information about the global state of the system. The only
decidable type of architectures is that of pipelines.
% This case can be
% solved with automata theoretic methods as the communication alphabet
% is fixed and the interaction of the process with the environment can
% be represented as a regular tree language. 

The setting we consider here is based on a by
now well-established model of distributed computation using
shared actions: \emph{Zielonka's asynchronous
  automata}~\cite{zie87}. Such a device is an asynchronous product of
finite-state processes synchronizing on common actions.  
Asynchronicity means that processes can progress at different speed. 
% This simple yet rich model has solid theoretical foundations rooted in the theory of Mazurkiewicz traces. 
Similarly to~\cite{GLZ04,MTY05} 
we consider the control problem for such automata. 
Given a Zielonka automaton (plant), find another Zielonka automaton
(controller) such that the product of the two satisfies a given
specification. In particular, the controller does
not restrict the parallelism of the system.
Moreover, during synchronization the individual processes of the controller
can exchange all their 
information about the global state of the system. This gives more
power to the controller than in the Pnueli and Rosner model, thus
avoiding simple pathological scenarios leading to undecidability. 
It is still open whether the control problem for Zielonka automata is
decidable.  

In this paper we prove decidability of the control problem  for
reachability objectives on tree architectures. 
In such architectures every process can communicate with its parent,
its children, and with the environment. 
If a controller exists, our algorithm yields a controller that is a
finite state Zielonka automaton exchanging information of
\emph{bounded} size.  We also provide the first non-trivial lower bound for
asynchronous distributed control. It matches the $l$-fold exponential
complexity of our algorithm ($l$ being the height of the
architecture). 
%Algorithms for other asynchronous subcases \cite{GLZ04,MTY05} are also non elementary, but without matching lower bounds. In the synchronous setting, 
%non-elementary lower bound is only claimed \cite{PnuRos89}.

As an example, our decidability result covers client-server
architectures where a server communicates with clients, and server and
clients have their own interactions with the environment
(cf. Figure~\ref{fig:server-client}). Our algorithm providing a
controller for this architecture runs in exponential time. Moreover,
each controller adds polynomially many bits to the state space of the
process.  Note also that this architecture is undecidable
for~\cite{PR90} (each process has inputs), and is not covered
by~\cite{GLZ04} (the action alphabet is not a co-graph), nor
by~\cite{MTY05} (there is no bound on the number of actions performed
concurrently).
\begin{figure}[b]
\centering
  \vspace{-0.7cm}
  \includegraphics[scale=.6]{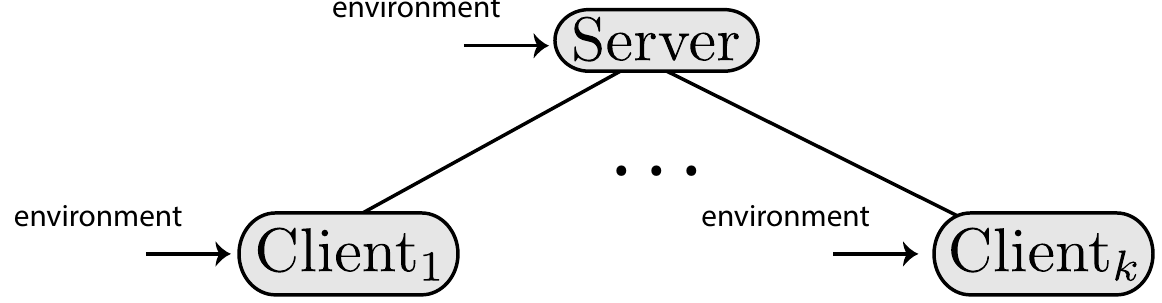}
%  \vspace{-0.7cm}
  \caption{Server/client architecture}
  \label{fig:server-client}
\end{figure}
\medskip

\noindent\textit{Related work.}
The setting proposed by Pnueli and Rosner~\cite{PR90} has been
thoroughly investigated in past years. 
By now we understand that, suitably using
the interplay 
between specifications and an architecture, one can get undecidability
results for most architectures rather easily. While
specifications leading to undecidability are very artificial, no
elegant solution to eliminate them exists at present.
% Results on multi-player
% games~\cite{PetRei79,PnuRos89} tell us that synthesis in this
% framework is undecidable, and~\cite{PR90} shows that synthesis
% w.r.t.~properties expressed in LTL is decidable when the %communication
% graph is a (directed) pipeline, with inputs allowed only at the first
% node.

The paper~\cite{kv01} gives an automata-theoretic approach to
solving pipeline architectures and at the same time extends the
decidability results to CTL$^*$ specifications and variations of the
pipeline architecture, like one-way ring architectures. The synthesis
setting is investigated in~\cite{MadThiag01} for local
specifications, meaning that each process has its own, linear-time
specification. For such specifications, it is shown that an
architecture has a decidable synthesis problem if and only if it is a
sub-architecture of a pipeline with inputs at both endpoints.
%For instance, the three process
%pipeline with inputs on the first two processes is undecidable. 
The paper~\cite{FinSch05} proposes information forks as an uniform
notion explaining the (un)decidability results in distributed
synthesis. 
In~\cite{MeyWil05} the authors consider distributed synthesis for
knowledge-based 
specifications. The paper~\cite{GasSznZei09} studies an interesting case
of external specifications and well-connected architectures.

Synthesis for asynchronous systems has been strongly advocated by Pnueli
and Rosner in \cite{PR89icalp}. Their notion of asynchronicity
is not exactly the same as ours: it means roughly that system/environment
interaction is not
turn-based, and processes observe the system only when
scheduled.
 This notion of asynchronicity appears in
several subsequent works, such as~\cite{sf06,KatSchPel11} for distributed
synthesis.

As mentioned above, we do not know whether the control problem in our setting is
decidable in general. Two related decidability results are known, both
of different flavor that ours. The first one~\cite{GLZ04} restricts the alphabet of actions: control with reachability
condition is decidable for co-graph alphabets.  This restriction 
excludes among others client-server architectures.
%is
%not satisfied as soon as there are local actions for each process, and a
%process that can communicate with two other ones (a case for which we
%show here that the control problem is decidable).  
The second result~\cite{MTY05} shows decidability by restricting the
plant: roughly speaking, the restriction says that every process
can have only bounded missing knowledge about the other processes (unless
they diverge). The proof
of~\cite{MTY05} goes beyond the controller synthesis problem, by
coding it into monadic second-order theory of event structures and
showing that this theory is decidable when the criterion on the plant
holds. Unfortunately, very simple plants have a  decidable control problem but
undecidable MSO-theory of the associated event structure.
Melli\`es~\cite{mel06} relates game semantics and asynchronous games, played
on event structures.
More recent 
work~\cite{cgw12} considers finite games on event structures
and shows a determinacy result for such games under some restrictions.

\medskip

 \noindent\textit{Organization of the paper.} 
 The next section presents basic definitions. 
 The two consecutive sections present the algorithm and the matching
 lower bound.

\section{Basic definitions and observations}
Our control problem can be formulated in the same way as the Ramadge
and Wonham control problem but using Zielonka automata instead of
standard finite automata. We start by presenting Zielonka automata and
an associated notion of concurrency. Then we briefly recall the
Ramadge and Wonham formulation and our variant of it. Finally, we give
a more convenient game-based formulation of the problem.

\subsection{Zielonka automata}

Zielonka automata are simple parallel devices. Such an
automaton is a parallel
composition of several finite automata, denoted as~\emph{processes},
synchronizing on common actions. There is no global clock, so between
two synchronizations, two processes can do a different number of
actions. Because of this Zielonka automata are also called
asynchronous automata.

A \emph{distributed action alphabet} on a finite set $\PP$ of processes is a
pair $(\S,\loc)$, where $\S$ is a finite set of \emph{actions} and
$\loc:\S \to (2^{\PP}\setminus \es)$ is a \emph{location
  function}. The location $\loc(a)$ of action $a \in\S$ comprises all
processes  that need to synchronize in order to perform this
action. 
A (deterministic) \emph{Zielonka automaton}
$\Aa=\struct{\set{S_p}_{p\in \PP},s_{in},\set{\d_a}_{a\in\S}}$ is
given by 
\begin{itemize}
\item for every process $p$ a finite set $S_p$ of (local) states,
\item the initial state $s_{in} \in \prod_{p \in \PP} S_p$, 
\item for every action $a \in\S$  a partial transition function
$\d_a:\prod_{p\in \loc(a)}S_p\stackrel{\cdot}{\to} \prod_{p\in \loc(a)}S_p$ on tuples of states of processes in
  $\loc(a)$. 
\end{itemize}

For convenience, we abbreviate a tuple $(s_p)_{p \in P}$ of local
states by  $s_P$,  where $P \subseteq \PP$. We also talk about $S_p$
as the set of \emph{$p$-states} and of $\prod_{p \in \PP} S_p$ as
\emph{global states}. Actions from $\S_p=\set{a \in\S \mid p
\in\loc(a)}$ are denoted as \emph{$p$-actions}.

A Zielonka automaton can be seen as a sequential automaton with the
state set $S = \prod_{p\in\PP} S_p$ and transitions $s \act{a} s'$ if
$(s_{\loc(a)}, s'_{\loc(a)}) \in \d_a$, and $s_{\PP\setminus
\loc(a)}=s'_{\PP\setminus \loc(a)}$. By $L(\Aa)$ we denote the set of
words labeling runs of this sequential automaton that start from the
initial state.

This definition has an important consequence. The location mapping
$\loc$ defines in a natural way an independence relation $I$: two
actions $a,b \in\S$ are independent (written as $(a,b)\in I$) if they
involve different processes, that is, if $\loc(a) \cap \loc(b)
=\emptyset$. Notice that the order of execution of two independent
actions $(a,b)\in I$ in a Zielonka automaton is irrelevant, they can
be executed as $a,b$, or $b,a$ - or even concurrently. More generally,
we can consider the congruence $\sim_I$ on $\S^*$ generated by $I$,
and observe that whenever $u \sim_I v$, the global state reached from
the initial state on $u$ and $v$, respectively, is the same. Hence, $u
\in L(\Aa)$ if and only if $v \in L(\Aa)$. Notice also that if $u
\sim_I vx$ and $x\in\S^*$ involves no $p$-action, then the $p$-state
reached on $u$ and $v$, respectively, is the same.

The idea of describing concurrency by an independence relation on
actions goes back to the late seventies, to Mazurkiewicz~\cite{maz77}
and Keller~\cite{kel73} (see also ~\cite{DieRoz95}). An equivalence
class $[w]_I$ of $\sim_I$ is called a Mazurkiewicz \emph{trace}, it can be
also viewed as labeled pomset of a special kind. Here, we will often
refer to a trace using just a word $w$ instead of writing $[w]_I$.  As
we have observed $L(\Aa)$ is a sum of such equivalence classes. In
other words it is \emph{trace-closed}.

\begin{example}
  Consider the following, very simple, example with processes
  $1,2,3$. Process 1 has local actions $a_0,a_1$ and synchronization
  actions $c_{i,j}$ ($i,j = 0,1$) shared with process 2. Similarly,
  process 3 has local actions $b_0,b_1$ and synchronization actions
  $d_{i,j}$ ($i,j = 0,1$) shared with process 2
  %blaise: switched and edited.
(cf.~Figure~\ref{f:async} where the symbol $*$ denotes any 
value 0 or 1).
%and $i,j,k,l$ each take both values).
Each process is a finite automaton and the Zielonka
automaton is the product of the three components synchronizing on
common actions. 
We have for instance $(a_i,b_j) \in I$ and $(c_{i,j},d_{k,l}) \notin
I$. The final states are the rightmost states of each automaton. The
automaton accepts traces of the form $a_ib_j c_{i,k} d_{j,l}$ with
$i=l$ or $j=k$.
\begin{figure}[htbp]
\centering
%blaise: delete bis cause i dont have the file!
  \includegraphics[scale=.6]{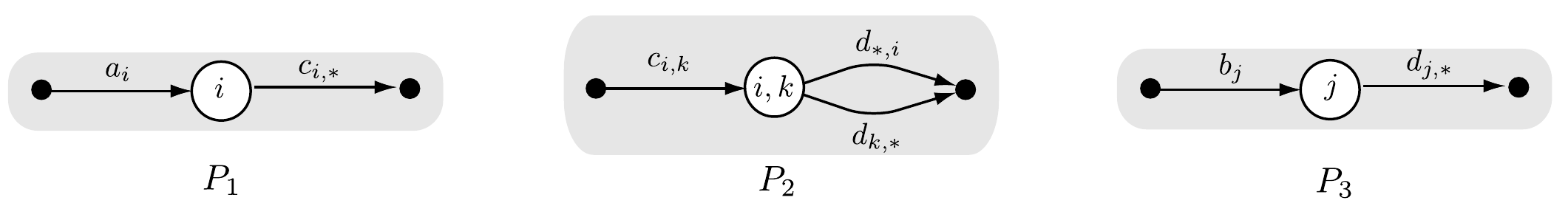}
  \caption{A Zielonka automaton}
  \label{f:async}
\end{figure}
\end{example} 
% For a detailed introduction to the theory of Mazurkiewicz traces, see
% the book~\cite{DieRoz95}. The theory of trace
% languages offers many results and tools. In particular, many
% fundamental results of the theory of regular languages have their
% equivalent trace versions. The most fundamental result is Zielonka's
% theorem below.

% A Zielonka automaton can be seen as a sequential automaton with
% the state set $S = \prod_{p\in\PP} S_p$ and transitions $s \act{a} s'$
% if $(s_{\loc(a)}, s'_{\loc(a)}) \in \d_a$, and $s_q=s'_q$ for all $q
% \notin \loc(a)$. By $L(\Aa)$ we denote the set of words labeling runs
% of this sequential automaton that start from the initial state.  This
% definition has an important consequence. If $(a,b) \in I$ then the
% same state is reached on  $ab$ and $ba$. More generally,
% whenever $u \sim_I v$ and $u \in L(\Aa)$ then $v \in L(\Aa)$,
% too. In other words, $L(\Aa)$ is
% \emph{trace-closed}.

Since the notion of a trace can be formulated without a reference to an
accepting device, it is natural to ask if the model of Zielonka
automata is powerful enough. Zielonka's theorem says that this is
indeed the case, hence these automata are a right model for
the simple view of concurrency captured by Mazurkiewicz traces.

\begin{theorem}\cite{zie87}
  Let $\loc :\S\to(2^\Pp\setminus\set{\es})$ be a distribution of
  letters. %, and $I$  the induced independence relation. 
  If a language $L\incl \S^*$ is regular and trace-closed then there is a
  deterministic Zielonka automaton accepting $L$ (of size exponential
  in the number of processes and polynomial in the size of the minimal
  automaton for $L$, see~\cite{ggmw10}). 
\end{theorem}

\medskip

One could try to use Zielonka's theorem directly to solve a
distributed control problem. For example, one can start with the
Ramadge and Wonham control problem, solve it, and if a solution happened to
respect the required independence, then distribute
it. Unfortunately, there is no reason for the solution to respect the
independence. Even worse, the following, relatively simple, result
says that it is algorithmically impossible to approximate a
regular language by a language respecting a given independence
relation.

\medskip

\begin{theorem}\cite{SEM03}
  It is not decidable if, given a distributed alphabet and a
  regular language $L\incl \S^*$, there is a trace-closed language
  $K \subseteq L$ such that every letter from $\S$ appears in some
  word of $K$.
\end{theorem}

\medskip

The condition on appearance of letters above is not crucial for the
above undecidability result. Observe
that we need some condition in order to make the problem nontrivial,
since by definition the empty language is trace-closed.

\subsection{The control problem}

We can now formulate our control problem as a variant of the 
Ramadge and Wonham formulation. We will then provide an equivalent
description of the problem in terms of games. While more complicated
to state, this description is easier to work with. 

Recall that in Ramadge and Wonham's control problem~\cite{RW89} we are
given an alphabet $\S$ of actions partitioned into system and
environment actions: $\S^{sys}\cup\S^{env}=\S$. Given a plant $P$ we
are asked to find a controller $C$ such that the product $P\times C$
satisfies a given specification. Here both the plant and the
controller are finite deterministic automata over $\S$. Additionally,
the controller is required not to block environment actions, which in
technical terms means that from every state of the controller there
should be a transition on every action from $\S^{env}$.

Our control problem can be formulated as follows: Given a distributed
alphabet $(\S,\loc)$ as above and a Zielonka automaton $P$,
find a Zielonka automaton $C$ over the same distributed alphabet such
that $P\times C$ satisfies a given specification. Additionally
the controller is required not to block uncontrollable actions: from
every state of $C$ every uncontrollable action  should be possible.
 The important point is
that the controller should have the same distributed structure as the
plant. The product of the two automata, that is just the standard
product, means that plant and controller are totally synchronized, in
particular communications between processes happen at the same
 time. 
Hence concurrency in the controlled system is the same as in
the plant. The major difference between the controlled system and the
plant is that the states carry the additional information computed by
the controller.

\begin{example}
  Reconsider the automaton in Figure~\ref{f:async} and assume that
  $a_i,b_j \in \Senv$ are uncontrollable. So the controller needs to
  propose controllable actions $c_{i,k}$ and $d_{j,l}$, resp., in such a way that all
  processes reach their final state. In particular, process 2 should
  not block. At first sight this may seem impossible to guarantee, as it
  looks like process $1$ needs to know what $b_j$ process $3$ has
  received, or process $3$ needs to know about the $a_i$ received by
  process $1$. Nevertheless, a controller exists. It consists of
  $P_1$ proposing $\set{c_{ii}}$ at state $i$, process $P_3$ proposing
  $\set{d_{j,1-j}}$ at state $j$ and process $P_2$ proposing all
  actions. If $i=j$ then $P_2$ reaches the final state by the
  transition $d_{k,*}$, else by the transition $d_{*,i}$.
  %choosing $k=i$ and $l=1-j$: if $i=j$ then $k=j$, else $i=l$.
\end{example}

It will be more convenient to work with a game formulation of this
problem. Instead of talking about controller we will talk about
distributed strategy in a game between \emph{system} and
\emph{environment}. A plant defines a game arena, with plays
corresponding to initial runs of $\Aa$. Since $\Aa$ is deterministic,
we can view a play as a word from $L(\Aa)$ - or a trace, since
$L(\Aa)$ is trace-closed.  Let
$\Plays(\Aa)$ denote the set of traces associated with words from
$L(\Aa)$.

A strategy for the system will be a collection of individual
strategies for each process. The important notion here is the view 
each process has about the global state of the system. Intuitively this is
the part of the current play that the process could see or  learn
about from other processes during a communication with them. Formally,
the $p$-view of a play $u$, denoted $\view_p(u)$, is the smallest
trace $[v]_I$ such that $u \sim_I vy$ and $y$ contains no action from
$\S_p$. We write $\Plays_p(\Aa)$ for the set of plays that are
$p$-views:
\[\Plays_p(\Aa)=\set{\view_p(u) \mid u\in\Plays(\Aa)}\,.
\]

A \emph{strategy for a process} $p$ is a function
$\s_p:\Plays_p(\Aa)\to 2^{\Ssys_p}$, where $\Ssys_p=\set{a \in \Ssys
\mid p\in\loc(a)}$.  We require in addition, for every $u \in
\Plays_p(\Aa)$, that $\s_p(u)$ is a subset of the actions that are
possible in the $p$-state reached on $u$. A \emph{strategy} is a
family of strategies $\set{\s_p}_{p\in \PP}$, one for each process.

The set of plays respecting a strategy $\s=\set{\s_p}_{p\in
  \PP}$, denoted $\Plays(\Aa,\s)$,  is the smallest set  
containing the empty play $\e$, and such that for every $u\in \Plays(\Aa,\s)$:
\begin{enumerate}
\item if $a\in \Senv$ and $u a \in \Plays(\Aa)$ then $u a$ is in
  $\Plays(\Aa,\s)$;
\item if $a\in \Ssys$ and $u a\in \Plays(\Aa)$ then $u
  a\in\Plays(\Aa,\s)$ provided that $a\in \s_p(\view_p(u))$ \emph{for all}
  $p\in\loc (a)$.
\end{enumerate}
Intuitively, the definition says that actions of the environment are
always possible, whereas actions of the system are possible only if
they are allowed by the strategies of all involved processes.
As in \cite{MTY05} (and unlike \cite{GLZ04}) our strategies
are process-based. That is, a controllable action $a$ with
$\loc(a)=\set{p,q}$ is allowed from $(s_p,s_q)$ if it is proposed by
process $p$ in state $s_p$ and by process $q$ in state $s_q$.
Before
defining winning strategies, we need to introduce infinite plays that
are consistent with a given strategy $\s$. Such plays can be seen as
(infinite) traces associated with infinite, initial runs of $\Aa$
satisfying the two conditions of the definition of
$\Plays(\Aa,\s)$. We write $\Plays^\infty(\Aa,\s)$ for the set of
finite or infinite such plays.  A play from $\Plays^\infty(\Aa,\s)$ is
also denoted as~\emph{$\s$-play}.

A play $u \in \Plays^\infty(\Aa,\s)$ is called
\emph{maximal}, if there is no action $c$ such that
$uc\in \Plays^\infty(\Aa,\s)$.  In particular, $u$ is maximal if
$\view_p(u)$ is infinite for every process $p$. Otherwise, if
$\view_p(u)$ is finite then $p$ cannot have enabled local actions
(either controllable or uncontrollable). Moreover there should be no
communication possible between any two processes with finite views in
$u$.

In this paper we consider \emph{local reachability} winning
conditions. For this, every process has a set of target states $F_p
\subseteq S_p$. We assume that states in $F_p$ are \emph{blocking},
that is they have no outgoing transitions. This means that if
$(s_{\loc(a)},s'_{\loc(a)}) \in \d_a$ then $s_p \notin F_p$ for all $p
\in \loc(a)$. 

\begin{definition}
The \emph{control problem} for a plant $\Aa$ and a local reachability
condition $(F_p)_{p \in \PP}$ is to determine if there is a strategy
$\s=(\s_p)_{p \in \PP}$ such that every maximal trace $u \in
\Plays^\infty(\Aa,\s)$ ends 
in $\prod_{p \in \PP} F_p$ (and is thus finite).  Such traces and
strategies are called \emph{winning}.
\end{definition}

As already mentioned, we do not know if this control problem is
decidable in general. In this paper we put one restriction on possible
communications between processes. First, we impose two
simplifying assumptions on the distributed alphabet $(\S,\loc)$. The 
first one is that all actions are at most binary: $|\loc(a)|\le 2$,
for every $a \in \S$. The second requires that all uncontrollable
actions are local: $|\loc(a)|=1$, for every $a \in
\Senv$. So the first restriction says that we allow only
binary synchronizations. It makes the technical reasoning much
simpler. The second restriction reflects the fact that each process is
modeled with its own, local environment.

\begin{definition}\label{df:com graph}
A distributed alphabet $(\S,\loc)$ with unary and binary actions
defines an undirected graph $\CG$ with node set $\PP$ and edges $\{p,q\}$
if there exists $a \in\S$ with $\loc(a)=\{p,q\}$, $p \not= q$. Such a
graph is called \emph{communication graph}. 
\end{definition}

% \medskip

% To sum up: in this paper we consider the Ramadge and Wonham
% formulation of the control problem but using Zielonka automata
% instead of standard ones. As specifications we consider reachability
% properties. We show that this problem is decidable for acyclic
% communication graphs (Theorem~\ref{th:main}). We also provide a tight complexity
% bound (Theorem~\ref{th:nonelem}).

%\input{upperbound.tex}

\section{The upper bound for acyclic communication graphs}

We fix in this section a distributed alphabet
$(\S,\loc)$. According to Definition~\ref{df:com graph} the alphabet
determines a communication graph $\CG$. We assume that $\CG$ is
acyclic and has at least one edge. This allows us to choose a leaf $r
\in \PP$ in $\CG$, with $\{q,r\}$ an edge in $\CG$. Throughout this
section, $r$ denotes this fixed leaf process and $q$ its parent
process. Starting from a
control problem with input $\Aa$, $(F_p)_{p \in\PP}$ we define below a
control problem over the smaller (acyclic) graph $\CG'=\CG_{\PP
  \setminus \{r\}}$. The construction will be an exponential-time
reduction from the control problem over $\CG$ to a control problem
over $\CG'$. If we represent $\CG$ as a tree of depth $l$ then
applying this construction iteratively we will get an $l$-fold
exponential algorithm to solve the control problem for $\CG$
architecture.

The main idea of the reduction is simple: process $q$ simulates
the behavior of process $r$. The reason why a simulation can work is
that after each synchronization between $q$ and $r$, the views of both
processes are identical, and between two such synchronizations $r$
evolves locally. But the construction is more delicate than this
simple description suggests, and needs some preliminary considerations
about winning strategies.

% \medskip
% \noindent\textit{Some preparatory lemmas:}
We start with a lemma showing how to
restrict the winning strategies. 
% The first one holds for
% arbitrary communication graphs, whereas the second one relies on the
% fact the $r$ is a leaf in $\CG$.
For $p,p' \in \PP$ let $\S_{p,p'} =\{a
\in\S \mid \loc(a)=\{p,p'\}\}$. So $\S_{p,p'}$ is the set of
synchronization actions between $p$ and $p'$. Moreover $\S_{p,p}$ is
just the set of local actions of $p$. We write $\S^{loc}_p$ instead of
$\S_{p,p}$ and $\S^{com}_p = \S_p \setminus \S^{loc}_p$. Recall that in
the lemma below $r$ is the fixed leaf process, and $q$ its
parent.

 \begin{lemma}\label{l:separation}
  If there exists some winning strategy for $\Aa$, then there is one, say
  $\s$,  such that for every  $u \in \Plays(\Aa,\s)$ the
  following hold:
  \begin{enumerate}
  \item If an uncontrolable action is possible from a state $s_r$ of
    process $r$ then for every play $u$ with $\state_r(u)=s_r$ we have
    $\s_r(\view_r(u))=\es$.
  \item For every process $p$ and $X=\s_p(\view_p(u))$, we have either
    $X=\set{a}$ for some $a\in\S^{loc}_p$ or $X\incl\S_p^{com}$.
  \item Let $X=\s_q(\view_q(u))$ with $X\incl\S_q^{com}$. Then either
    $X \subseteq \S_{q,r}$ or $X\incl\S_q^{com} \setminus \S_{q,r}$
    holds.
  \end{enumerate}
\end{lemma}

\begin{proof}
  The first item is immediate, since uncontrollable actions are alwyas
  possible. For the second item we modify $\s$ into $\s'$ as
  follows. If $\s_p(u)$ contains some local action, then we choose
  one, say $a$, and put $\s'_p(u)=\set{a}$. We do this for every
  process $p$ and show that the resulting strategy $\s'$ is
  winning. Suppose that $v \in \Plays(\Aa,\s')$ is maximal, but not
  winning. Clearly $v$ is a $\s$-play, but not a maximal one, since
  $\s$ is winning.  Thus, there is $vc \in \Plays(\Aa,\s)$ for some
  processes $p \not=p'$ and some $c \in \S_{p,p'}$. By definition of
  $\s'$ it means that either $\s_p(\view_p(v))$ or
  $\s_{p'}(\view_{p'}(v))$ contains some local action, say $a \in
  \s_p(\view_p(v))$ and $\s'_p(\view_p(v))=\set{a}$. But then $va$ is
  a $\s'$-play, a contradiction with the maximality of $v$.

For the last item we can assume that $\s_q$ and $\s_r$ always
  propose either a local action or a set of communication actions. Now
  given a winning strategy $\s$ we will produce a winning strategy
  $\s'$ satisfying the condition of the lemma, by modifying only $\s_r$.

 Assume that $u \in \Plays_q(\Aa,\s)$ with $s_q=\state_q(u)$, and
 $\s_q(u)=B \cup C$, where $B 
 \subseteq \S_{q,r}$ and $C \subseteq \S_q^{com} \setminus \S_{q,r}$
 with both $B,C$  non-empty.
 We define $\s'_q$ by cases:
 \begin{equation*}
   \s'_q(u)=\begin{cases}
     C & \text{there exists $(s_r,A)\in\Sync^{\s}_r(u)$ with
       $(s_r,A) \bowtie (s_q,B)=\es$,}\\
     B & \text{otherwise}.
   \end{cases}
 \end{equation*}
The idea behind the definition above is simple: if there is a possible
local future for $r$ that makes synchronization with $q$ impossible (first
case), then $q$'s strategy can as well propose only communication with
other processes than $r$ -- since such communication leads to winning
as well. If not, $q$'s strategy can  offer only communication with $r$,
since this choice will never block.

We show now that $\s'$ is winning.  Assume by contradiction that $v$
is a maximal $\s'$-play, but not winning.  It is then a $\s$-play, but
not a maximal one. So there must be some $a \in \S_q^{com}$ such that
$va \in \Plays(\Aa,\s)$.  In particular, $q$'s state after $v$ is not
final.  Let $u=\view_q(v)$, $s_q=\state_q(u)$, and $\s_q(u)=B\cup C$
with $B\incl \S_{q,r}$ and $C\incl \S^{com}_q\setminus \S_{q,r}$. We
have two cases.

Suppose $\s'_q(u)=C$, so we are in the first case of the above
above. Thus there exists $(s_r,A)\in\Sync^\s_r(u)$ such that $(s_r,A)
\bowtie (s_q,B) =\es$. By definition of $\Sync^\s_r$ we find
$x\in(\S^{loc}_r)^*$ such that $u'=ux$ is a $\s$-play and
$\s_r(\view_r(u'))=A$. Since $u=\view_q(u')$, we have
$\s_q(\view_q(u'))=B\cup C$. This means that no communication between
$q$ and $r$ is possible after $u'$. No local action of $q$ is possible
after $u'$ since $u=\view_q(v)$, and we have assumed that $v$ is a
maximal $\s'$-play. Finally, by the choice of $x$, no local action of
process $r$ is possible from $u'$. To obtain a contradiction it
suffices to show that $u'$ can be extended to a maximal $\s$-play by
adding a sequence of actions $w$ of processes other than $q$ and
$r$. This will do as $\state_q(u')$ is not accepting by assumption,
and we will get a maximal $\s$-play that is not winning.  To find the
desired $w$ observe that $v \sim uwy$ where
$w\in(\S\setminus (\S_q\cup\S_r))^*$ and $y\in\S^*_r$. So $y$
represents the actions of $r$ after the last action of  $q$ in $v$,
and $w$ represents the actions of other processes.  Taking $v'=uwx$ we
observe that $v' \sim u'w$ and 
that $v'$ is a maximal $\s$-play. So we have found the desired $w$.

The second case is when $\s'_q(u)=B$. This means that for all 
$(s_r,A)\in\Sync^\s_\ell(u)$, we have 
$(s_r,A) \bowtie (s_r,B) \neq \es$. Since $v$ is a maximal $\s'$-play,
 no local action of $r$ is possible. This means that
 $(s_r,A):=(\state_r(v),\s_r(\view_r(v)))\in\Sync^\s_r(u)$. But then
 $(s_r,A) \bowtie (s_q,B)\not=\es$. Since $\s'_q(u)=B$ there is
 some  possible communication between
 $q$ and $r$ after $v$, so $v$ is not maximal
 w.r.t.~$\s'$.
\end{proof}

The following definition associates with a strategy $\s$ and the leaf
process $r$ all the outcomes of local plays of $r$ such that $r$ is
either waiting for a synchronization with $q$ or is in a final (hence
blocking) state.  For an initial run $u$ of $\Aa$ we denote by
$\state_p(u)$ the $p$-state reached by $\Aa$ on $u$.

 \begin{definition}\label{d:outcome}
Given a strategy $\s$ and a
$\s$-play $u$, let $\Sync_r^\s(u)\subseteq S_r \times \Pp(\S_{q,r})$ be the set:
\begin{align*}
  \Sync_r^\s(u)=\{(s_r,A) \mid \ & \exists x\in(\S^{loc}_r)^*\, .\
    \text{$ux$ is a $\s$-play,}\\
      & \state_r(ux)=s_r,\, 
     \s_r(\view_r(ux))=A\incl \S_{q,r}, \text{and}\\
     &s_r \text{ final or } A\not=\es\}\,.
\end{align*}   
 \end{definition}
Observe that if $\s$ allows $r$ to reach a final state $s_r$ from $u$
without communication, then $(s_r,\es)\in \Sync_r^\s(u)$. This is so,
since final states are assumed to be blocking.

For the game reduction we need to precalculate all possible sets
$\Sync^\s_r$. These sets will be actually of the special form
described below.

\begin{definition}
  Let $s_r$ be a state of $r$.  We say that $T\incl
  S_r\times\Pp(\S_{qr})$ is an \emph{admissible plan in $s_r$} if
  there is a play $u$ with $\state_r(u)=s_r$, and a strategy $\s$ such
  that (i) $T=\Sync^\s_r(u)$, (ii) every $\s$-play of $r$ from $s_r$
  reaches a final state or a state where $\s$ proposes some communication
  action, and (iii) one of the following holds:
  \begin{itemize}
  \item $A\not=\es$ for every $(t_r,A)\in T$, or
    \item $t_r \in F_r$ and $A=\es$ for every $(t_r,A)\in T$.
  \end{itemize}
  In the second case $T$ is called a \emph{final plan}.
\end{definition}

It is not difficult to see that we can compute the set of all
admissible plans. In the above definition we do not ask that $\s$ is
winning in the global game, but just that it can locally bring $r$ to one
of the situations described by $T$. So verifying if $T$ is an
admisible plan simply amounts  to solve a 2-players reachability game on process
$r$ against the (local) environment.

Lemma~\ref{l:rfuture} below allows to deduce that the sets $\Sync_r^\s$
are admissible plans whenever $\s$ is winning. For $(s_r,A), (s_q,B)$ with $s_q \in S_q, s_r \in
S_r$, $A,B \subseteq \S_{q,r}$ let $(s_r,A) \bowtie (s_q,B):=\set{a
  \in A \cap B \mid \d_a(s_q,s_r) \text{ is defined}}$. So $(s_r,A)
\bowtie (s_q,B)$ contains all actions belonging to both $A$ and $B$,
that are enabled in the state $(s_q,s_r)$.

\begin{lemma}\label{l:rfuture}
  If $\s$ is a winning strategy 
  satisfying Lemma~\ref{l:separation} then for every $\s$-play $u$ in $\Aa$ we have:
  \begin{enumerate}
  \item if there is some $\s$-play $uy$ with $y\in(\S \setminus \S_r)^*$ and
    $\state_q(uy)\in F_q$  then $\Sync^\s_r(u)$ is a final plan;
  \item if there is some $\s$-play $uy$ with $y\in(\S \setminus
    \S_r)^*$, $s_q=\state_q(uy)$, $\s_q(uy)=B\incl \S_{q,r}$, and
    $B\not=\es$ then for every $(t_r,A)\in\Sync^\s_r(u)$ we have
    $(s_q,B) \bowtie (t_r, A)\not=\es$.
  \end{enumerate}
 In particular,  $\Sync^\s_r(u)$ is always an admissible plan.
\end{lemma}

\begin{proof}
  Take $y$ as in the statement of the lemma and suppose $\state_q(uy)
  \in F_q$. Take $(t_r,A)\in \Sync^\s_r(u)$. By definition this means
  that there is $x\in (\S^{loc}_r)^*$ such that $ux$ is a $\s$-play,
  $\state_r(ux)=t_r$, and $\s_r(\view_r(ux))=A$ with $A\incl
  \S_{q,r}$. Observe that $uyx$ is also a $\s$-play. Hence $t_r$ should
  be final because after $uyx$ process $r$ can do at most
  communication with $q$, but this is impossible since $q$ is in a
  final state. Since $t_r$ is final, it cannot propose an action,
  hence $A=\es$. This shows the first item of the lemma.

  For the second item of the lemma take $y$, $s_q$, $B$, and $(t_r,A)$ as in
  the assumption. Once again we get $x\in (\S^{loc}_r)^*$ such that
  $ux$ is a $\s$-play, $\state_r(ux)=t_r$, and $\s_r(\view_r(ux))=A$
  with $A\incl \S_{q,r}$. Once again $uyx$ is a $\s$-play. We have
  that $s_q$ is not final since $B\not=\es$. As
  $\s$ is winning, the play $uyx$ can be extended by an action of
  $q$. But the only such action that is possible is a communication
  between $q$ and $r$. Since $A$ and $B$ are the communication sets
  proposed by $\s_r$ and $\s_q$, respectively, we must
  have $(s_q,B) \bowtie (t_r,A) \not=\es$.
\end{proof}

\medskip
\textit{The new plant $\Aa'$}.  We are now ready to define the reduced
plant $\Aa'$ that is the result of eliminating process $r$.  Let
$\PP'=\PP\setminus \set{r}$.  We have $\Aa'=\struct{\set{S'_p}_{p\in
    \PP'},s'_{in},\set{\d'_a}_{a\in\S'}}$ where the components will be
defined below.

The states of process $q$ in $\Aa'$ are of one of the following types:
\begin{equation*}
\struct{s_q,s_r}\,,\quad \struct{s_q,T}\,,\quad  \struct{s_q,T,B}  \,,
\end{equation*}
where $s_q \in S_q, s_r\in S_r$, $T \incl S_r\times\Pp(\S_{q,r})$ is
an admissible plan, $B\subseteq \S_{q,r}$. The new initial state for $q$ is
$\struct{(s_{in})_q,(s_{in})_r}$.

For every $p \not= q$, we let $S'_p=S_p$ and $F'_p =F_p$. The local
winning condition for $q$ becomes 
%\begin{equation*}
$F'_q = F_q \times F_r \cup
\set{\struct{s_q,T} \mid s_q \in F_q, \text{ and $T$  is a final plan}}$.
%\end{equation*}

The set of actions $\S'$ is $\S \setminus \S_r$, plus
additional local $q$-actions that we introduce below.
All transitions $\d_a$ with $\loc(a) \cap \{q,r\}=\es$ are as in $\Aa$.
Regarding $q$ we have the following transitions:
\begin{enumerate}
\item If not in a final state then process $q$ chooses an admissible
  plan:
  \begin{equation*}
\struct{s_q,s_r} \act{\ch(T)} \struct{s_q,T},    
  \end{equation*}
where $T$ is an admissible plan in $s_r$, and
  $\struct{s_q,s_r} \notin F_q \times F_r$.

\item Local action of $q$:
\[
\struct{s_q,T} \act{a} \struct{s'_q,T}, \quad\text{ if } s_q
\act{a} s'_q \text{ in } \Aa\,.
\]
\item Synchronization between $q$ and $p \not= r$:
\[(\struct{s_q,T},s_p) \act{b} (\struct{s'_q,T},s'_p), \quad\text{if }
  (s_q,s_p) \act{b} (s'_q,s'_p)\, .\]
\item Synchronization between $q$ and $r$. Process $q$ declares the
  communication actions with $r$:
  \begin{equation*}
    \struct{s_q,T} \act{\ch(B)} \struct{s_q,T,B},\qquad \text{ if }\;  
     B \subseteq \S_{q,r}\,
  \end{equation*}
  when $s_q$ is not final,  $T$ is not a final plan, and for every
  $(t_r,A)\in T$ we have $(t_r,A) \bowtie  (s_q,B)\not=\es$.

 %  \igw{Attention: before we had this}[[such that $A \cap T(t_r) \not=\es$
%   for each $t_r \in \loc(T)$ if $T \in \Tt^{com}$. We require in
%   addition that either $s_q \notin F_q$ or $T \notin \Tt^{fin}$.]]

  Then the environment can choose the target state of $r$ and a
  synchronization action $a \in \S_{q,r}$:
  \begin{equation*}
\struct{s_q,T,B} \act{(a,t_r)} \struct{s'_q,s'_r}\qquad \text{if
  $(s_q,t_r) \act{a} (s'_q,s'_r)$ in $\Aa$}
  \end{equation*}
  for every $(a,t_r)$ such that $(t_r,A)\in T$ for some $A$, and $a\in
  A\cap B$. Notice that the
complicated name of the action $(a,t_r)$ is needed to ensure
that the transition is deterministic.

\end{enumerate}

To summarize the new actions of process $q$ in plant $\Aa'$ are:
\begin{itemize}
\item $\ch(T) \in \Ssys$, for every admissible plan  $T$,
\item $\ch(B) \in \Ssys$, for each $B \subseteq \S_{q,r}$,
\item $(a,t_r) \in \Senv$ for each
  $a \in \S_{q,r},t_r \in S_r$.
\end{itemize}

The proof showing that this construction is correct provides a
translation from winning strategies in $\Aa$ to winning strategies in
$\Aa'$, and back. To this purpose we rely on a
translation from plays in $\Aa$ to plays in $\Aa'$. A (finite or
infinite) play $u$ in $\Aa$ is a trace that will be convenient
to view as a word of the form
\begin{equation*}
u = y_0 x_0a_1\    \cdots\ a_i y_i x_i\ a_{i+1}\dots
\end{equation*}
where for $i\in \Nat$ we have that: $a_i \in \S_{q,r}$ is communication between $q$ and $r$; $x_i
\in (\S^{loc}_r)^*$ is a sequence of local actions of $r$; and $y_i\in (\S
\setminus \S_r)^*$ is a sequence of actions of other processes than $r$. Note that $x_i,y_i$ are
concurrent, for each $i$. We will write $u|_{a_i}$ for the prefix of
$u$ ending in $a_i$. Similarly $u|_{y_i}$ for the prefix ending with 
$y_i$; analogously for $x_i$.

Fix a strategy $\s$ in $\Aa$. With a word $u$ as above we will associate the word
%\begin{multline*}
\[  \chi(u)=\ch(T_0)y_0\ch(B_0)(a_1,t^1_r)\, \cdots \,
(a_i,t^i_r)\ch(T_i)\ y_i\ch(B_i)(a_{i+1},t^{i+1}_r)\dots
\]
%\end{multline*}
where for every $i=0,1,\dots$:
\begin{itemize}
\item $T_i=\Sync^\s_r(u|_{a_i})$ and $T_0=\Sync^\s_r(\e)$;
\item $B_i=\s_q(\view_q(u|_{y_i}))$;
\item $t^i_r=\state_r(u|_{x_i})$.
\end{itemize}
We then construct a strategy that plays $\chi(u)$ in $\Aa'$ instead of
$u$ in $\Aa$. In Figure~\ref{fig:chi} we have pictorially represented
which parts of $u$ determine which parts of $\chi(u)$.

\begin{figure*}[tbhf]
  \centering
\begin{tikzpicture}[xscale=1.5]
  \draw (0,1)node{$u =$}
  ++(1.3,0)node(y0){$y_0$} ++(1.2,0)node(x0){$x_0$}
  ++(.5,0)node(a1){$a_1$} ++(1.8,0)node(y1){$y_1$}
  ++(1.3,0)node(x1){$x_1$} ++(.6,0)node(a2){$a_2$};
  \draw (0,0)node{$\chi(u) =$} 
  ++(.8,0)node(T0){$\ch(T_0)$}
  ++(.5,0)node(y0p){$y_0$} ++(.7,0)node(A0){$\ch(B_0)$}
  ++(1,0)node(a1t){$(a_1,t^1_r)$} ++(.8,0)node(T1){$\ch(T_1)$}
  ++(1,0)node(y1p){$y_1$}
  ++(.7,0)node(A1){$\ch(A_1)$}
   ++(1.2,0)node(a2t){$(a_2,t^2_r)$}
   ++(1,0)node(T2){$\ch(T_2)$};
    \draw(y0)--(y0p);
    \draw(y0)--(A0);
    \draw(x0)--(a1t);
    \draw(a1)--(a1t);
    \draw(a1)--(T1);
   \draw(y1)--(y1p);
   \draw(y1)--(A1);
   \draw(x1)--(a2t);
   \draw(a2)--(a2t);
    \draw(a2)--(T2);
\end{tikzpicture}
  \caption{Definition of $\chi(u)$}
  \label{fig:chi}
\end{figure*}
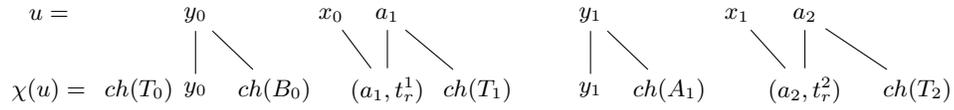

The next lemma follows directly from the definition of the reduction
from $\Aa$ to $\Aa'$.

\medskip

\begin{lemma}\label{lemma:states}
  If $u$ ends in a letter from $\S_{q,r}$ then we have the following
  \begin{itemize}
  \item $\state_q(\chi(u))=\struct{\state_q(u),\state_r(u)}$.
  \item $\state_p(\chi(u)y)=\state_p(uy)$ for every $p\not=q$ and
    $y\in (\S \setminus \S_{q,r})^*$.
  \item $\state_q(\chi(u)\ch(T)y)=\struct{\state_q(uy),T}$ for every $y\in (\S \setminus \S_{q,r})^*$.
  \item $\state_q(\chi(u)\ch(T)y\ch(B))=\struct{\state_q(uy),T,B}$ for every $y\in (\S \setminus \S_{q,r})^*$.
  \end{itemize}
\end{lemma}

\medskip
\noindent\textit{From $\s$ in $\Aa$ to $\s'$ in $\Aa'$.}
We are now ready to define $\s'$ from a winning strategy $\s$. We
assume that $\s$ satisfies the property stated in
Lemma~\ref{l:separation}. We will define $\s'$ only for certain plays and
then show that this is sufficient. 

Consider $u'$ such that $u'=\chi(u)$ for some $\s$-play $u$ ending in
a letter from $\S_{q,r}$. We have:
\begin{itemize}
\item If $\state_q(u') \notin F_q$ then $\s'_q(\view_q(u'))=\set{\ch(T)}$
  where $T=\Sync^\s_r(u)$.
\item For every process $p\not=q$ we put
  $\s'_p(\view_p(u'\ch(T)y))=\s_p(\view_p(uy))$ for $y\in
  (\S \setminus \S_{q,r})^*$.
\item For $y\in (\S
  \setminus \S_{q,r})^*$ and $B=\s_q(\view_q(uy))$ we define
  \begin{equation*}
    \s'_q(\view_q(u'\ch(T)y))=
    \begin{cases}
      B & \text{if $B\cap \S_{q,r}=\es$} \\
      \set{\ch(B)} & \text{if $B\incl \S_{q,r}$}  
    \end{cases}
  \end{equation*}
\item $\s_q'(\view_q(u'\ch(T)y\ch(B)))=\es$.
\end{itemize}
Observe that in the last case the strategy proposes no move as there
are only moves of the environment from a position reached on a play of
this form.

The next lemma states the correctness of the construction.

 \begin{lemma}\label{lemma:sigma to sigma prim}
  If $\s$ is a winning strategy for $\Aa,(F_p)_{p \in\PP}$ then $\s'$
  is a winning strategy for $\Aa',(F'_p)_{p \in \PP'}$.
\end{lemma}

\begin{proof}
  We will show inductively that for every $\s'$-play $u'$ ending 
  in a letter of the form $(a',t'_r)$ there is a $\s$-play
  $u$ such that $u'=\chi(u)$. Then we will show that
  every maximal $\s'$-play is winning.
  
  We start with the induction step, later we will explain how to
  do the induction base. Let us take $u'=\chi(u)$ as in the induction
  hypothesis. By Lemma~\ref{lemma:states} we have
  $\state_q(u')=\struct{\state_q(u),\state_r(u)}$.   

  Consider a possible, $\s'$-compatible, extension of $u'$ till the next letter
  $(a,t_r)$. It is of the form $u'\ch(T)y\ch(B)(a,t_r)$ where $y\in
  (\S \setminus \S_r)^*$.  We will show that it is of the form
  $\chi(uyxa)$ for some $x\in (\S_r^{loc})^*$,
  and that $uyxa$ is a $\s$-play.
  \begin{itemize}
  \item By definition of the automaton $\Aa'$ and the strategy $\s'$
    we have $\s'(u')=\set{\ch(T)}$ with $T=\Sync^\s_r(u)$.
  \item Since $\s'$ is the same as $\s$ on actions from $\S \setminus
    \S_r$, we get that $uy$ is a $\s$-play.

  \item Concerning $\ch(B)$, by the definition of $\s'$ we have that
    $B=\s_q(\view_q(uy))$. Then by the definition of $\Aa'$ we get
    some $A$ such that $(t_r,A)\in T$, and $a\in (t_r,A)\bowtie (s_q,B)$
    with $s_q=\state_q(uy)$. As $T=\Sync^\s_r(u)$ we can find $x\in
    (\S^{loc}_r)^*$ such that $ux$ is a $\s$-play, $\state_r(ux)=t_r$
    and $\s_r(ux)=A$. We get that $uyxa$ is a $\s$-play with
    $\chi(uyxa)=u'\ch(T)y\ch(B)(a,t_r)$, and we are done.
  \end{itemize}
The induction base is exactly the same as the induction step 
taking $u'$ and $u$ to be the empty sequence. 

To finish the lemma we need to show that every maximal $\s'$-play 
is winning. For this we examine all possible situations where
such a play can end. We consider plays $u'$ and $u$ as at the
beginning of the lemma.

If $u'$ itself is maximal then $\state_q(u')$ is final because
otherwise $\ch(T)$ would be possible. Hence, by
Lemma~\ref{lemma:states} $\state_q(u)$ and $\state_r(u)$ are
final. Since $\s$ and $\s'$ are the same on processes other than $q$
and $r$, no action $a$ with $\loc(a)\cap \set{q,r}=\es$ is possible
from $u$. It follows that $u$ is a maximal $\s$-play. Since $\s$ is
winning, $\state_p(u)$ is final for every process $p$. By
Lemma~\ref{lemma:states}, $u'$ is winning too.

Suppose now that $u'\ch(T)y$ is maximal for some
$y\in(\S\setminus\S_{r})^*$. By the same reasoning as above there is no
$\s$-play extending $uy$ by an action from $\S \setminus\S_r$. We have two cases
\begin{itemize}
\item If $\state_q(uy)$ is final then $T$ is a final plan by
  Lemma~\ref{l:rfuture}. So there is $x\in(\S^{loc}_r)^*$ such that
  $\state_r(uyx)$ is final. Then $uyx$ is a maximal $\s$-play. Since
  $\s$ is winning, after $uyx$ all processes are in the final state. By
  Lemma~\ref{lemma:states}, $u'\ch(T)y$ is winning too.
\item If $state_q(uy)$ is not final then $\s(uy)\incl
  \S_{q,r}\not=\es$ since $\s$ is assumed to satisfy
  Lemma~\ref{l:separation}, and communication with other processes
  than $r$ is not possible. By Lemma~\ref{l:rfuture} $T$ cannot be
  final and action $\ch(B)$ for $B=\s(uy)$ is possible according to
  $\s'$. A contradiction.
\end{itemize}

A play of the form $u'ch(T)y\ch(B)$ cannot be maximal since some local
actions of the form $(a,t_r)$ are always possible. This covers all the
cases and completes the proof.
\end{proof}

\medskip
\noindent\textit{From $\s'$ in $\Aa'$ to $\s$ in $\Aa$.}
From a strategy $\s'=(\s'_p)_{p \in \PP'}$ for $\Aa'$ we define a
strategy $\s=(\s_p)_{p \in \PP}$ for $\Aa$. We assume that $\s'$
satisfies Lemma~\ref{l:separation}. We consider $u$ ending in an
action from $\S_{q,r}$ such that $\chi(u)$ is a $\s'$-play. First, for
every $p\not=q,r$ and every $y\in (\S \setminus \S_r)^*$ we set
\begin{equation*}
  \s_p(\view_p(uy))=\s'_p(\view_p(\chi(u)y)).
\end{equation*}
If $\state_q(\chi(u))$ is not final then $\s'(\chi(u))=\set{\ch(T)}$
for some admissible plan $T$ in state $\state_r(\chi(u))$. This means
that $T=\Sync^\r_r(u)$ for some strategy $\r$. In this case:
\begin{itemize}
\item for every $x\in(\S^{loc}_r)^*$ we set $\s_r(ux)=\r_r(ux)$;

\item for every $y\in (\S \setminus \S_r)^*$ we consider
  $X=\s'_q(\view_q(\chi(u)\ch(T)y))$ and set
  \begin{equation*}
    \s_q(\view_q(uy))=
    \begin{cases}
      B & \text{if $X=\set{\ch(B)}$}\\
      X & \text{otherwise}
    \end{cases}
  \end{equation*}
\end{itemize}

\medskip

\begin{lemma}\label{lemma:sigma prim to sigma}
  If $\s'$ is a winning strategy for $\Aa',(F'_p)_{p \in\PP'}$ then $\s$
  is a winning strategy for $\Aa,(F_p)_{p \in \PP}$.
\end{lemma}

\medskip

\begin{proof}
  Suppose that $u$ is $\s$-play ending in an action from
  $\S_{q,r}$ and such that $\chi(u)$ is a $\s'$-play. We first show
  that for every extension of $u$ to a $\s$-play $uyxa$ with $y\in
  (\S \setminus \S_r)^*$, $x\in(\S^{loc}_r)^*$, and $a\in \S_{q,r}$, its image
  $\chi(uyxa)$ is a $\s'$-play. Then we will show that
  every maximal $\s$-play is winning.

  Take $uyxa$. By Lemma~\ref{lemma:states} $\state_q(\chi(u))$ is not
  final, so we have $\s'(\chi(u))=\set{\ch(T)}$. Then
  $T=\Sync^\s_r(u)$ by the definition of $\s$. Again directly from the
  definition we have that $\chi(u)\ch(T)y$ is a $\s'$-play. By
  definition of $\s$ we have then that $\chi(u)\ch(T)y\ch(B)$ is a
  $\s'$-play for $B=\s_q(\view_q(uy))$. Finally, we need to see why
  $(a,t_r)$ with $t_r=\state_r(ux)$ is possible. Since
  $T=\Sync^\s_r(u)$ we get that $(t_r,\s_r(\view_r(ux)))\in T$. Then
  $a\in\s_r(\view_r(ux))\cap B$, and in consequence
  $\chi(u)\ch(T)y\ch(B)(a,t_r)$ is possible by
  Lemma~\ref{lemma:states} and the definition of $\Aa'$.

  It remains to verify that every maximal $\s$-play is winning. Consider a
  maximal $\s$-play $uyx$ where $u$ ends in an action from $\S_{q,r}$,
  $x\in(\S^{loc}_r)^*$, and $y\in(\S \setminus\S_r)^*$ (this includes the cases
  when $x$, or $y$ are empty). We look at $\chi(u)$ and consider two
  situations:

  \begin{itemize}
  \item If no $\ch(T)$ is possible from $\chi(u)$ then
    $\state_q(\chi(u))$ is final. This means that $x$ is empty and
    $\state_q(u)$ and $\state_r(u)$ are both final. It is then clear
    that $\chi(u)y$ is a maximal $\s'$-play. Since $\s'$ is winning,
    every process is in a final state. So $uy$ is a winning play in
    $\Aa$.
  \item If $\chi(u)\ch(T)$ is a $\s'$-play for some $T$ then again we
    have two cases:
    \begin{itemize}
    \item If $s_r=\state_r(uyx)$ is final then $(s_r,\es)\in T$ by
      the definition of $\s$. As $T$ is an admissible plan, $T$ is
      final. After $\chi(u)y$ no action other than $\ch(B)$ is
      possible. But $\ch(B)$ is not possible either since $T$ is
      final. Hence $\chi(u)y$ is a maximal $\s'$-play. So all the
      states reached on $\chi(u)y$ are final. By
      Lemma~\ref{lemma:states} we deduce the same for $uyx$, hence
      $uyx$ is winning.
    \item If $s_r$ is not final then $\s_r(\view_r(uyx))=A\incl \S_{q,r}$ for
      $A\not=\es$ (local actions of $r$ are not possible, since $uyx$
      is maximal). Hence $(s_r,A)\in T$, and $T$ is not final. This
      means that $s_q=\state_q(\chi(u)\ch(T)y)$ is not final. So it is
      possible to extend the $\s'$-play with an action of the form $\ch(B)$. But
      by the definition of $\Aa'$ we have $(s_q,B) \bowtie
      (s_r,A)\not=\es$. Hence 
      $uyx$ can be extended by a communication between $q$ and $r$ on
      a letter from $B\cap A$; a contradiction.
    \end{itemize}
  \end{itemize}
\end{proof}

\medskip

Together, Lemmas~\ref{lemma:sigma to sigma prim} and~\ref{lemma:sigma prim to
  sigma} show Theorem~\ref{th:reduction}.

\begin{theorem}\label{th:reduction}
  Let $r$ be the fixed leaf process with $\PP'=\PP\setminus \set{r}$
  and $q$ its parent.  Then the system has a winning strategy for $\Aa,(F_p)_{p\in\PP}$ iff it
  has one for $\Aa',(F'_p)_{p\in\PP'}$. All the components of $\Aa'$
  are identical to those of $\Aa$, apart that for the process $q$. The
  size of $q$ in $\Aa'$ is $\Oo(M_q2^{M_r |\S_{qr}|})$, where $M_q$ and
  $M_r$ are the sizes of processes $q$ and $r$ in $\Aa$, respectively.
\end{theorem}

\begin{remark} Note that the bound on $|\Aa'|$ is better than $|\Aa| +
\Oo(M_r 2^{M_\ell 2^{|\S_{r\ell}|}})$ obtained by simply counting all
possible states in the description above. 
The reason is that we can restrict admissible plans to be (partial)
functions from $S_\ell$ into $\Pp(\S_{r,\ell})$. 
That is, we do not need to consider different
sets of communication actions for the same state in $S_\ell$.
\end{remark}

Let us reconsider the example from Figure~\ref{fig:server-client} of a
server with $k$ clients. Applying our reduction $k$ times we reduce
out all the clients and obtain the single process plant whose size is
$M_s2^{(M_1+\dots +M_k)c}$ where $M_s$
is the size of the server, $M_i$ is the size of client $i$, and $c$ is
the maximal number of communication actions between a client and the
server.

\begin{theorem}\label{th:main}
  The control problem for distributed alphabets with acyclic communication
  graph is decidable. There is an algorithm for solving
  the problem (and computing a finite-state controller, if it exists) whose
  working time is bounded by a tower of exponentials of height equal
  to half of the diameter of the graph. 
\end{theorem}

Our reduction algorithm  can be actually used to compute a (finite-state)
distributed controller:

\begin{corollary}
  There is an algorithm which solves the control problem for
  distributed alphabets whose communication graph is acyclic and if
  the answer is positive, the algorithm outputs a controller
  satisfying the following property: For every process $p$ and every
  state $s$ of the controller $\Aa_c$, the set of actions allowed for
  process $p$ in state $s$ is the set of all uncontrollable local
  actions plus:
\begin{itemize}
  \item either a unique controllable local action,
  \item or a set of controllable actions shared with a unique neighbour $q$
  of $p$.
\end{itemize} 
\end{corollary}

\section{The lower bound}

We show in this section that in the simplest non-trivial case of
acyclic communication graphs, consisting of a line of three processes,
the control problem is already \EXPTIME-complete. In the general case
the complexity of the control problem grows as a tower of exponentials
function with respect to the size of the diameter of the communication
graph.

\subsection{Height one}
\begin{proposition} \label{p:3} 
  The control problem for the communication graph $1 \edge 2 \edge 3$ is
  \EXPTIME-complete.
\end{proposition}

\medskip

\begin{proof}
%blaise edit
  The EXPTIME upper bound follows from Theorem~\ref{th:reduction}, as
  the height of the tree is 1.
So the reduction is applied twice
from process $2$, first simulating process $1$, then simulating
process $3$. Finally, a reachability game is solved on an exponential
size arena.  

  For the lower bound we simulate an alternating polynomial space
  Turing machine $M$ on input $w$. We assume that $M$ has a unique
  accepting, blocking configuration (say with blank tape, head
  leftmost). The goal now is to let processes $1,3$ guess an accepting
  computation tree of $M$ on $w$. The environment will be able to
  choose a branch in this tree and challenge each proposed
  configuration. Process $2$ will be used to validate tests initiated
  by the environment. If a test reveals an inconsistency, process 2
  blocks and the environment wins. To summarize the idea of the construction: \anca{added}
  processes $1$ and $3$
  generate sequences of configurations (encoded by local actions),
  separated by action $\$$ and $\bar{\$}$, respectively, shared with
  process $2$. Both start with the initial configuration of $M$ on
  $w$. Transitions from existential states are chosen by the plant,
  and those from universal ones by the environment. 
At a given time, process $1$ has generated the same number of
configurations is process $3$, or process $3$ is about generating one
configuration more. In the first case, the environment can check
that it is the
same configuration; and in the second, it can check that it is the successor
configuration. In this way, $1$ and $3$ need to generate the same branch of the
run tree.

  A computation of $M$ with space bound $n$ is a sequence $C_0
  \vdash C_1 \vdash \cdots \vdash C_N$, where each configuration $C_i$ is
  encoded as a word from $\G^* (Q \times \G) \G^*$ of length
  $n$. Since $M$ is alternating, its acceptance is expressed by the existence
  of a tree of accepting computations. 

  Processes 1 starts by generating the initial configuration on $w$, followed by a synchronization
  symbol $\$$ with process 2. After this, process 1 generates a sequence of
  configurations separated by $\$$. When generating a
  configuration, process 1 remembers $M$'s state $q$ and the symbol $A$
  under the head. All transitions so far are controllable. After
  generating $\$$ process 1 goes into a state where the outgoing
  transitions are labeled by $M$'s transitions on $(q,A)$ (if the
  configuration was not blocking). These transitions are controllable
  if $q$ is existential, and uncontrollable if $q$ is
  universal. The transition chosen, either by the plant or the
  environment, is  stored in the state up to the next synchronization symbol.  Finally, if the current
  configuration is final then process 1 synchronizes with 2 on $\$_F$
  (instead of $\$$) and goes into an accepting state. % All actions so
%   far are local, except for $\$,\$_F$ which are shared with process 2 (see
%   also Figure~\ref{f:3}).

  The description is similar for process 3, with
  $\bar{\G},\bar{Q},\bar{\$}, \bar{\$}_F$ instead of $\G,Q,\$,\$_F$.  Finally, process
  $2$ has two main states, $\mathit{eq}$ and $\mathit{succ}$, with transitions $\mathit{eq}
  \act{\bar\$} \mathit{succ}$ and $\mathit{succ} \act{\$} \mathit{eq}$. From state $\mathit{eq}$ it can
  go to an accepting state after reading $\bar{\$}_F$ followed by $\$_F$.

\begin{figure}[htb]
  \centering
\begin{tikzpicture}[yscale=1.3,xscale=1.2]
\path (0,1)node(a){} (0,0)node(ap){} (2.2,1)node(b){}
(2.2,0)node(bp){} (4.4,1) node(e) {} (4.4,0) node(ep) {};
\path (1,1.2) node(c) {$C_0$} (1,-0.2) node(c') {$\bar{C}_0$} 
(3.2,1.2) node(d) {$C_1$} (3.2,-0.2) node(d') {$\bar{C}_1$} 
(5,1.2) node (f) {$C_2$} (5,-0.2) node (fp) {$\bar{C}_2$}
%(6.2,1.1) node (g) {$\a$} (6.2,-0.2) node (gp) {$\b$} 
(6.7,.9) node (h) {$(i,\a)$} (6.5,0.1)  node (hp) {$\bar{(j,\b)}$} 
(1.8,.25) node (ip) {$\bar{\$}$} (2.1,.75) node (i) {$\$$}
(4,.25) node (jp) {$\bar{\$}$} (4.3,.75) node (j) {$\$$}

;
\path (-.2,1) node {1}  (-.2,0) node {3} (-.2,.5) node {2};
\draw (a) -- ++(2.2,0) -- +(0,-0.5);
\draw (b) -- ++(2.2,0) -- +(0,-0.5);
\draw (e) -- ++(1.9,0) -- +(0,-0.5);

\draw (ap) -- ++(2,0) -- +(0,0.5);
\draw (bp) -- ++(2,0) -- +(0,0.5);
\draw (ep) -- ++(1.7,0) -- +(0,.5);
\end{tikzpicture}
\caption{Environment chooses positions $i,j$ in
    $C_P,\bar{C}_P$ with $P=2$. System wins iff $\a=\b$ or $i
    \not= j$.}
  \label{f:3}
\end{figure}
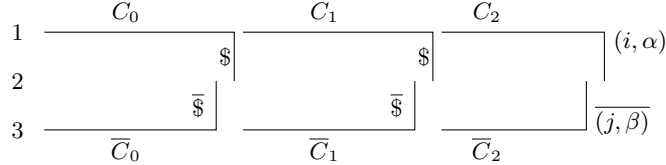

\vspace{-.5cm}

The environment can initiate 2 kinds of tests: equality and successor
test. The equality test checks that $C_P=\bar{C}_P$ and the successor
test checks that $C_P \vdash \bar{C}_{P+1}$. 

For the equality test, the environment can choose a position $i$
within $C_P$ and a position $j$ in $\bar{C}_P$. Formally, for each
(controllable) outgoing transition $s \act{\a}$ of process 1 with $\a
\in \G \cup (Q \times \G)$ there is a transition $s \act{(\eqtest,\a)}
(\eqtest,i,\a)$ with $(\eqtest,\a)$ uncontrollable. The target state
$(\eqtest,i,\a)$ records the tape position $i$ (known from $s$) and
the tape symbol $\a$.  In state $(\eqtest,i,\a)$ process 1
synchronizes with 2 on action $(\eqtest,i,\a)$, and then stops
(accepting). The same for process 3 with uncontrollable actions
$\bar{(\eqtest,\b)}$, and synchronization action $\bar{(\eqtest,j,\b)}$.

From state $\mathit{eq}$ process $2$ can perform a synchronization
$\bar{(\eqtest,j,\b)}$ with process $3$ and then one with process $1$ on
any $(\eqtest,i,\a)$, provided $i\not=j$ or $\a=\b$, and then
accept. This is the case where the environment has chosen positions on
both lines 1 and 3 (see Figure~\ref{f:3}). If the environment has
chosen a test transition in $C_P$ but not in $\bar{C}_P$ (or
vice-versa), process 2 will accept (and stop), too.The successor test
is similar. 

The successor test is similar, it consists in choosing a position
within $C_P$ and one within $\bar{C}_{P+1}$. The information checked
by process 2 includes the symbols $\a_-,\a,\a_+$ of
$C_P$ at positions $i-1,i,i+1$ resp., so process 1 goes on transition
$\bar{(\inctest,\a)}$ into a state of the form $(i,\a,\a_-,\a_+)$. In
state $\bar{t}$ process 2 can perform a synchronization on
$(\inctest,i,\a,\a_-,\a_+)$ with process 1, and then one with process
3 on $\bar{(\inctest,j,\b)}$, provided $i \not= j$ or the symbols
$\a_-,\a,\a_+$ are inconsistent with the new middle symbol $\b$
according to $M$'s transition relation. 

The reader may notice that \anca{added} we need to guarantee that the universal
transitions chosen by the environment are the same, for processes $1$
and $3$. This can be enforced by communicating the transitions with actions
$\$,\bar{\$}$ to process $2$, who is in charge of 
checking. Moreover, note that the action alphabet above is not constant,
in particular it depends on $n$. This can be fixed by replacing each
action of type $(\eqtest,i,\a)$ (or alike) by a sequence of
synchronization actions where $i$ is transmitted bitwise. By
alternating the bits transmitted by $1$ and $3$, respectively, process
2 can still compare indices $i,j$.

Note also that configurations $C_P, \bar{C}_P$ are generated in
parallel, and so are $C_P$ and $\bar{C}_{P+1}$. This is crucial
for the correctness.
\end{proof}

\begin{lemma}
  The control problem defined in Proposition~\ref{p:3}  has a winning strategy if and only
  if $M$ accepts $w$. 
\end{lemma}

\medskip

\begin{proof}
  We assume that there is a winning strategy in the control game. Let us consider a maximal winning play without
  tests, where process 1 generates $C_0 \$ C_1 \$ \cdots C_N \$_F$ and
  process 3 generates $\bar{C}_0 \bar\$\bar{C}_1 \bar\$ \cdots
  \bar{C}_{N'} \bar{\$}_F$. By construction, each of the $C_p$ and
  $\bar{C}_q$ are configurations of length $n$, $C_0=\bar{C}_0$ is
  the initial configuration of $M$ on $w$, and $C_N=\bar{C}_{N'}$ is
  the accepting configuration. Suppose by contradiction
  that $C_0, \ldots, C_N$ is not a run of $M$. Assume first
  that $C_p=\bar{C}_p$ for all $0 \le p < P$, but $C_{P-1}
  \not\vdash \bar{C}_P$. In this case the environment could have chosen the
  first position $i$ where $\bar{C}_P$ does not correspond to a
  successor of $C_{P-1}$, and process 2 would have rejected after the
  synchronization $(\inctest,i,\a,\a_-,\a_+)$ followed by
  $\bar{(\inctest,i,\b)}$, contradicting the fact that the strategy is
  winning. The second case is where $C_p=\bar{C}_p$ for all $0 \le p <
  P$, but $C_P \not= \bar{C}_P$.  Then the environment could have chosen the
  first position $i$ where $C_P$ and $\bar{C}_P$ differ, and process 2
  would have rejected after the synchronization $\bar{(\eqtest,i,\b)}$
  followed by $(\eqtest,i,\a)$ with $\a\not= \b$, again a
  contradiction.  This means that $C_0 \vdash C_1 \vdash \cdots
  C_N$. Moreover, $C_N=\bar{C}_N$ is final since process 1 is in a
  final state (thus also $N=N')$.

  For the converse, we assume that $M$ accepts $w$. Let the
  strategy of processes 1 and 3 consist of generating an accepting run
  tree of $M$ on $w$.  For existential configurations, say that both 1
  and 3 choose the first winning transition among all
  possibilities. Every maximal play without environment test corresponds
  to an accepting run $C_0 \vdash C_1 \vdash \cdots C_N$, hence the
  play reaches a final state on every process. Every maximal play with
  test is of one of the following forms: (1) $C_0 \bar{C}_0 \bar\$ \$
  \cdots C_{P-1} \bar{C}_{P-1}\bar\$ \$ xy$, where $x$ and $y$ are
  prefixes of $C_P$ and $\bar{C}_P$, followed by $\eqtest$-actions, or
  (2) $C_0 \bar{C}_0 \bar\$ \$ \cdots \bar{C}_{P-1} \bar\$ x y$, where
  $x$ is prefix of $C_{P-1}$ and $y$ a prefix of $\bar{C}_P$, followed
  by $\inctest$-actions. In both cases, the environment's challenge
  fails, since $C_P=\bar{C}_P$ and $C_{P-1} \vdash \bar{C}_P$.
\end{proof}

\subsection{Lower bound: general case}

Our main objective now is to show how using a communication
architecture of diameter $l$ one can code a counter able to represent numbers
of size $\Tower(2,l)$ (with $\Tower(n,l)=2^{\Tower(n,l-1)}$ and
$\Tower(n,1)=n$). Then an easy adaptation of the construction will allow
to code computations of Turing machines with the same space bound as
the capabilities of counters.

We fix $n$ and will be first interested to define $n$-counters.  Let
$\S_i=\set{a_i,b_i}$ for $i=1,\dots,n$. We will think of $a_i$ as $0$
and $b_i$ as $1$, mnemonically: $0$ is round and $1$ is tall. Let
$\S^\#_i=\S_i\cup\set{\#_i}$ be the alphabet extended with an end
marker.

A $1$-counter is just a letter from $\S_1$ followed by $\#_1$. The value
of $a_1$ is  $0$, and the one of $b_1$ is $1$. Following this
intuition we write $(1-c)$ to denote $b$ if $c=a$ and vice versa.

An \emph{$(l+1)$-counter} is a word
\begin{equation}\label{eq:ctr}
  x_0u_0  x_1u_1\cdots  x_{k-1} u_{k-1}\#_{l+1}
\end{equation}
where $k=\Tower(2,l)$ and for every $i$, letter $x_i\in\S_{l+1}$ and
$u_i$ is an $l$-counter with value $i$. The value of the above
$(l+1)$-counter is $\sum_{i=0,\dots,k} x_i2^i$.  The end marker
$\#_{l+1}$ will be convenient in the construction that follows.
An \emph{iterated $(l+1)$-counter} is a nonempty sequence of
$(l+1)$-counters.

\medskip For every $l$ we will define a plant $\Cc^l$ such that the winning
strategy for the system in $\Cc^l$ will need to produce an iterated
$l$-counter. 

\medskip For $l=1$ this is very easy, we have only one
process in $\Cc^1$ and all transitions are controllable. 

\noindent\centerline{\includegraphics[scale=.62]{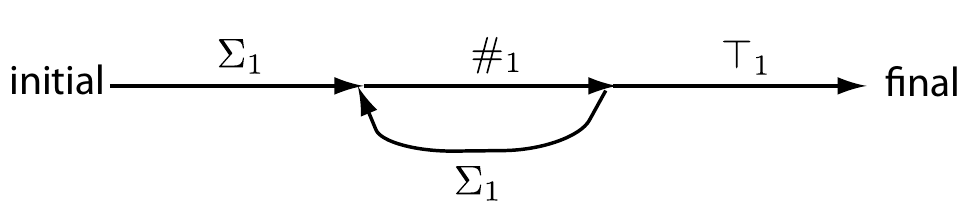}}

This automaton can repeatedly produce a $1$-counter and eventually go
to the accepting state. The letter on which it goes to accepting state
will be not important, so we put $\top_1$. Recall that our acceptance
condition is that all processes reach a final state from which no
actions are possible.

Suppose that we have already constructed $\Cc^l$. We want now to
define $\Cc^{l+1}$, a plant producing an iterated $(l+1)$-counter,
i.e., a sequence of $l$-counters with values $0,1,\ldots, (\Tower(2,l)-1),
0,1, \ldots$. We assume that the
communication graph of $\Cc^l$ has the distinguished root process $r_l$. Process
$r_l$ is in charge of generating an iterated $l$-counter. From $\Cc^l$ we
will construct two plants $\Dd^l$ and $\bar \Dd^l$, over disjoint sets
of processes. The plant $\Dd^l$ is obtained by adding a new root
process $r_{l+1}$ that communicates with $r_l$, similarly for the 
plant $\bar \Dd^l$ with root process $\bar{r_{l+1}}$.  The plant
$\Cc^{l+1}$ will be the composition of $\Dd^l$ and $\bar \Dd^l$ with a new
\emph{verifier process} that we name $\Vv_{l+1}$. The root process of the communication graph
of $\Cc^{l+1}$ will be $r_{l+1}$. The schema of the
construction is presented in Figure~\ref{fig:architecture}. Process
$r_{l+1}$, as well as $\bar{r_{l+1}}$, are in charge of generating an
iterated $(l+1)$-counter. That they behave indeed this way is
guaranteed by a construction similar to the one of
Proposition~\ref{p:3}, with the help of the verifier $\Vv_{l+1}$:
\anca{added} the
environment gets a chance of challenging each $l$-counter of the
sequence of $r_{l+1}$ (and similarly for $\bar{r_{l+1}}$). These
challenges correspond to two types of tests, equality and
successor. If there is an error in one of these sequences then 
the environment can place a challenge and win. Conversely, if there is
no error no challenge of the environment can be successful; this means
then that the sequences of $l$-counters have
correct values $0,1,\ldots,(\Tower(2,l)-1),0,1,\ldots$.

\begin{figure}[htbp]
  \centering
  \includegraphics[scale=.7]{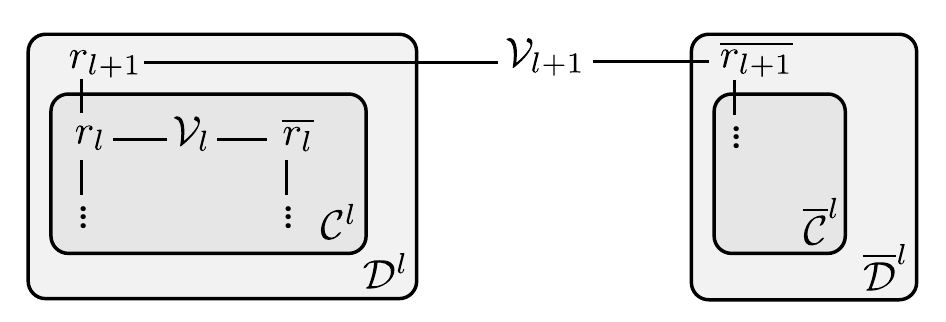}
  \caption{Architecture of the plant $\Cc^{l+1}$}
  \label{fig:architecture}
\end{figure}

\medskip\noindent\textit{Construction of $\Dd^l$.}
The construction of the automaton  of the new root
$r_{l+1}$ is presented in Figure~\ref{fig:main automaton}.

\begin{figure*}[htbp]
  \centering
  \includegraphics[scale=.7]{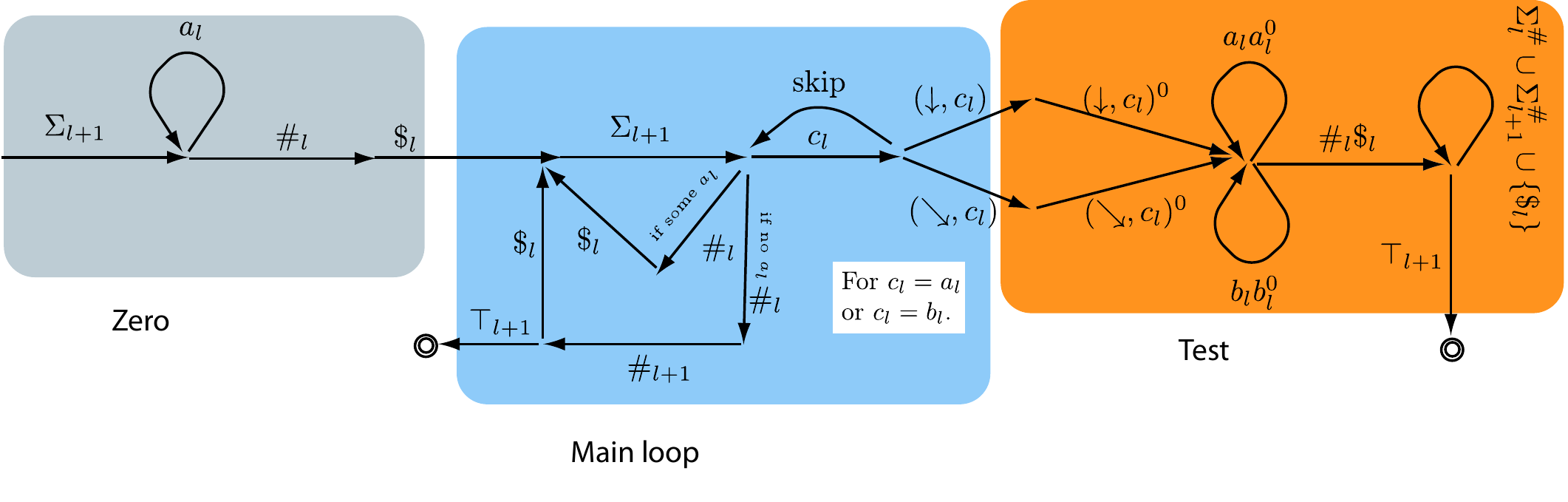}
  \caption{Automaton for process $r_{l+1}$}
  \label{fig:main automaton}
\end{figure*}

We start by modifying the automaton for process $r_l$, given by $\Cc^l$.
Actions of $r_l$ from $\S^\#_l$, that were previously local for
$r_l$, become shared actions with $r_{l+1}$.  Process $r_{l+1}$ has new local actions 
$\S^\#_{l+1}$ and an action $\$_l$, shared with process $\Vv_{l+1}$. The action
$\$_l$ is executed after each $l$-counter, that is, after each $\#_l$.

The automaton for $r_{l+1}$ has two main tasks: it ``copies'' the
sequence of $l$-counters generated by $r_l$ (actually only the
projection onto $\S_l$) and it interacts with $\Vv_{l+1}$ towards the
verification of this sequence. This automaton is composed of three parts that synchronize
with $r_l$, forcing it to behave in  some specific way. The first part called
``zero'' enforces that $r_l$ starts with an $l$-counter with value 0
(otherwise $r_{l+1}$ would block). When we read $\#_l$ we know that
the first $l$-counter has ended and the control is passed to the
second, main part of $r_{l+1}$.

The main part of $r_{l+1}$ gives a possibility for the environment
to enter into a test part. That is, after each transition on $c_l\in\S_l$
(that is $a_l$ or $b_l$) the environment chooses between action
$\sskip$ (that continues the main part) or a test action from
$\set{(\eqtest,c_l), (\inctest,c_l)}$  that leads into the test
part. The main part also outputs a local action $\#_{l+1}$ when
needed, i.e., whenever the last seen $l$-counter was
maximal. (Technically it means that there has been no $a_l$ since the
last $\#_l$.) The transition on $\#_{l+1}$ gives a possibility to go
to the accepting state.

The test part of $r_{l+1}$ simply receives the $\S_l$-actions of
$r_l$ and sends them to process $\Vv_{l+1}$ (cf.~loop $a_la_l^0$ and $b_lb_l^0$). It does so until it receives
$\#_l$ signaling the end of the counter. Then it sends $\$_l$ to
process $\Vv_{l+1}$ to inform it that the counter has finished. After this
$r_{l+1}$ enters in a state where it can do any controllable
action. From this state at any moment it can enter the accepting
state on a dummy letter $\top_{l+1}$.

\noindent\textit{Plant $\bar\Dd^l$.}
This one is constructed in almost the same way as $\Dd^l$. Most
importantly all processes (and actions) in $\bar\Dd^l$ are made disjoint from
$\Dd^l$. We will write $\bar a$ for the letter of $\bar\Dd^l$
corresponding to $a$ in $\Dd^l$.

The other difference between $\Dd^l$ and $\bar\Dd^l$ is that in the
latter every transition $\bar{(\inctest,c)}$ is changed into
$\bar{(\inctest,1-c)}$ if since the last $\$_l$ there have been only
$\bar{l_l}$. This is done to accommodate for the carry needed for the successor test. Recall that
$(1-c)$ stands for $a$ if $c$ is $b$  and vice versa.

\noindent\textit{Process $\Vv_{l+1}$.}
This process will have two main states $\mathit{eq}$ and $\mathit{succ}$, the first one
being initial. From $\mathit{eq}$ there is a transition on $\bar\$_l$ to $\mathit{succ}$,
and from $\mathit{succ}$ there is a transition on $\$_l$ back to
$\mathit{eq}$. Moreover from $\mathit{eq}$ it is possible to go to the accepting state.

\begin{figure}
  \centering
  \includegraphics[scale=.62]{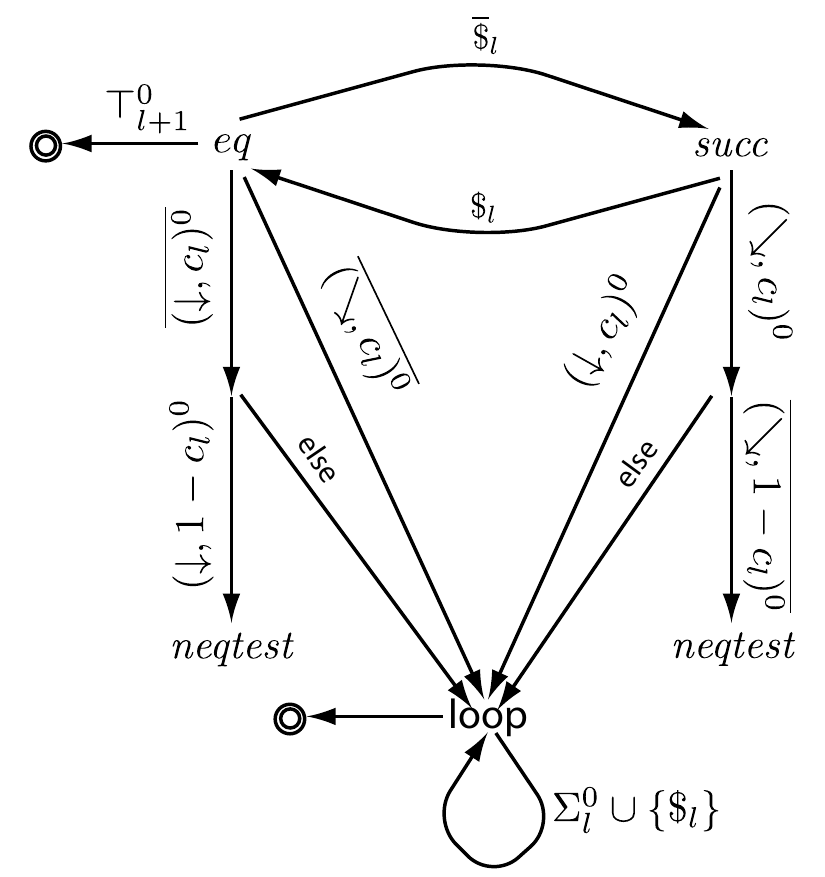}
  \caption{Process $\Vv_{l+1}$.}
  \label{fig:verifier}
\end{figure}

Additionally, from $\mathit{eq}$ there is a transition on $\bar{(\eqtest,c)^0}$
to the state $(eq,c)$ for every $c\in\S_l$. Similar to the
construction of Proposition~\ref{p:3}, process $\Vv_{l+1}$ should
accept if either the two bits from $\S_l$ challenged by the
environment are compatible with the test, or their positions
are unequal. So, from state $(eq,c)$ on letter
$(\eqtest,1-c)^0$ there is a transition to a state called
$\mathit{neqtest}$; on all other letters there is a transition to a
looping state (see also Figure~\ref{fig:verifier}). Similarly from $\mathit{succ}$, but now with $(\inctest,c)$
letters, and the order of reading from the components reversed.

From state $\mathit{neqtest}$ process $\Vv_{l+1}$ verifies that the sequence of
actions $\S^0_l$ initiated by $r_{l+1}$ has not the same \emph{length} as the
sequence over $\bar \S^0_l$ initiated by $\bar{r_{l+1}}$ (up to the
moment where $\$^0_l$ and $\bar\$^0_l$ are executed). This is done
simply by interleaving the two sequences of actions $a^0_l,b^0_l$, shared with
$r_{l+1}$ and $\bar{r_{l+1}}$, respectively\anca{deja la}. Notice that
the symbols $a^0_l,b^0_l$ by themselves are not important, one could as well replace
them by a single symbol. If this is the
case, then process $\Vv_{l+1}$ gets to an accepting state, otherwise it
rejects. In state $\mathit{loop}$ process $\Vv_{l+1}$ can perform any 
controllable action and then enter the
accepting state.

\medskip\noindent\textit{Putting together $\Cc^{l+1}$.}
The plant $\Cc^{l+1}$ is the composition of $\Dd^l$, $\bar\Dd^l$ and the
new process $\Vv_{l+1}$. The actions of $\Cc^{l+1}$ are the ones of $\Cc^l$,
plus $X \cup \bar X$ where $X$ consists of:
\begin{itemize}
\item $\S^\#_{l+1}\subseteq \Ssys$ with domain $\{r_{l+1}\}$,
\item $\S^\#_l\subseteq \Ssys$  with domain $\set{r_l,r_{l+1}}$,
\item $\sskip \in \Senv$ and $(\eqtest,c),
  (\inctest,c) \in \Senv$ with domain $\set{r_{l+1}}$ ($c \in
  \S_l$),
\item $c^0$, $\$_l$, $(\eqtest,c)^0$, and $(\inctest,c)^0$, all in $\Ssys$
  with domain $\set{r_{l+1},\Vv_{l+1}}$ ($c\in \S_l$).
\end{itemize}
The set $\bar X$ is defined similarly, by replacing every action $c$
by $\bar c$, and $r_l,r_{l+1}$ by $\bar{r_l},\bar{r_{l+1}}$ in the
domain of the action.

First we show that the system can indeed win every control instance
$\Cc^l$. Moreover he can win and produce at the same time any iterated
$l$-counter.

\medskip

\begin{lemma}\label{lemma:strategy-exits}
  For every level $l$ and every iterated $l$-counter $\cc$ there is a
  winning strategy $\s$ in $\Cc^l$ such that for every $\s$-play
  the projection of this play on
  $\bigcup_{i=1,\ldots,l} \S^\#_i$ is $\cc$.
\end{lemma}

\medskip

\begin{proof}
  The proof is by induction on $l$. For $l=1$ this is obvious since
  there are no environment moves and all possible behaviours leading
  to the accepting state are iterated $1$-counters.

  Let us consider level $l+1$. Recall that $\Cc^{l+1}$ is constructed
  from $\Cc^l$, $\bar{\Cc}^l$, and three new processes: $\rlp$,
  $\brlp$, $\Vv_{l+1}$. Fix an iterated $(l+1)$-counter $\cc$. Observe that
  the projection of $\cc$ on the alphabet of $l$-counters, namely
  $\bigcup_{i=1,\dots,l} \S^\#_i$, is an iterated $l$-counter.  By
  induction we have a winning strategy producing this
  counter in $\Cc^l$. We play this winning strategy in the $\Cc^l$ and
  $\bar\Cc^l$ parts of $\Cc^{l+1}$. It remains to say what the new
  processes  should do.

  Process $r_{l+1}$ should just produce $\cc$. By induction assumption
  we know that the letters this process reads from $r_l$ are the
  projection of $\cc$ on the alphabet of the $l$-counter; and it is so no
  matter if there are environment questions in $\Cc^l$ or not. So
  process $\rlp$ has to just fill in missing $\S_{l+1}$ letters. If
  the environment asks no questions to $r_{l+1}$ then at the end of
  $\cc$, this process will do $\#_{l+1}$, then $\top_{l+1}$ and enter the accepting
  state. Analogously for $\bar{r_{l+1}}$. At the same time process $\Vv_{l+1}$
  will be at state $\mathit{eq}$ and it can enter the accepting state, too,
  since it can count how many $\$_l$ symbols he has received.

  Let us suppose now that the environment chooses a question action in
  $r_{l+1}$ or $\bar{r_{l+1}}$. Let $i$ be the index of an $l$-counter $u_i$
  within $\cc$ at which the first question is asked. We will consider
  two cases: (i) the question is asked in $\bar{r_{l+1}}$, (ii) the
  question is asked in $r_{l+1}$ but not in $\bar{r_{l+1}}$.

  If a question is asked in $\bar{r_{l+1}}$ then the play has the
  following form:
  \begin{center}
  \begin{tikzpicture}[yscale=.7]
    \draw (0,2)node{$\rlp$:}++(1.5,0)node{$\dots u_{i-1}$}++(1.2,0)node(dol){$\$_l$}
    ++(1,0)node{$u$}++(1,0)node(d){$d$};
    \draw (0,1)node{$\Vv_{l+1}$:};
    \draw
    (0,0)node{$\brlp$:}++(1.5,0)node{$\dots\bar{u_{i-1}}$}++(1,0)node(dolb){$\bar
      \$_l$}++(1,0)node{$\bar v$}++(1,0)node(e){$\bar e$};
    \draw (dol)--+(0,-1);
    \draw (dolb)--+(0,1);
    \draw (d)--+(0,-1);
    \draw (e)--+(0,1);
  \end{tikzpicture}
  \end{center}
  \noindent with $u,\bar v$ being prefixes of $u_i$; $e$ being a
  question, and $d$ a synchronization action of $\rlp$ with $\Vv_{l+1}$. So
  $d$ can be a question or $\$_l$. Observe that after reading
  $\bar{\$_l}\$_l$ process $\Vv_{l+1}$ is in the state $\mathit{eq}$. It means that if
  the sequence $\bar e d$ is not $\bar{(\dar,c)}(\dar,1-c)$ for some
  $c\in\S_l$ then $\Vv_{l+1}$ enters state $\mathit{loop}$. From there it can
  calculate how many inputs from $r_{l+1}$ and $\bar{r_{l+1}}$ it is
  going to receive. It receives them and then enters the accepting
  state.  If $\bar e d$ is $\bar{(\dar,c)}(\dar,1-c)$ then $\Vv_{l+1}$ enters
  state $\mathit{neqtest}$. Since $\rlp$ and $\brlp$ output the same iterated
  counter it must be that the questions are placed in different
  positions of the two counters. But then $\Vv_{l+1}$ will receive from the
  two processes a different number of $\S_l$ letters. Hence it will
  enter eventually into the accepting state also in this case.

  Process $\brlp$ after receiving a question moves to a test component
  where it transmits the remaining part of the $l$-counter to $\Vv_{l+1}$
  followed by $\bar{\$_l}$. Then it enters into the loop state of the
  test copy and can continue to generate $\cc$ since it can do any
  transition in this state. As for process $\rlp$, if $d$ is a
  question, then it does the same thing as $\brlp$.  If $d$ is $\$_l$
  then $\rlp$ can continue to produce $\cc$, and both $\Vv_{l+1}$ and $\brlp$
  can simulate their behaviour as if no question has occurred. If the
  environment asks a question to $\rlp$ at some moment, it too will
  enter into accepting state and continue to produce $\cc$.

  If the first counter with a question is in $\rlp$ but not in $\brlp$
  then the play has the form:
  \begin{center}
  \begin{tikzpicture}[yscale=.7]
    \draw (0,2)node{$\rlp$:}++(1.5,0)node{$\dots u_{i-1}$}++(1.2,0)node(dol){$\$_l$}
    ++(1,0)node{$u$}++(1,0)node(d){$d$};
    \draw (0,1)node{$\Vv_{l+1}$:};
    \draw
    (0,0)node{$\brlp$:}++(1.5,0)node{$\dots\bar{u_{i-1}}$}++(1,0)node(dolb){$\bar
      \$_l$}++(1,0)node{$\bar u_i$} ++(1,0)node(dolbp){$\bar\$_l$}++(1,0)
    node{$\bar v$}++(1,0)node(e){$\bar e$}; 
    \draw (dol)--+(0,-1);
    \draw (dolb)--+(0,1);
    \draw (dolbp)--+(0,1);
    \draw (d)--+(0,-1);
    \draw (e)--+(0,1);
  \end{tikzpicture}
  \end{center}
  where $u$ is a prefix of $u_i$, $\bar v$ a prefix of $u_{i+1}$, $d$
  is a question, and $\bar e$ a synchronization of $\brlp$ with
  $\Vv_{l+1}$. Observe that after reading $\$_l\bar{\$_l}$ process $\Vv_{l+1}$ is in
  state $\mathit{succ}$. As before our first goal is to show that $\Vv_{l+1}$
  gets to an accepting state. If the sequence $d\bar e$ is not
  $(\inctest,c_l)\bar{(\inctest,1-c_l)}$ then we reason as in the
  previous case. Otherwise $\Vv_{l+1}$ gets to state $\mathit{neqtest}$. As before we
  can deduce that the two questions are asked at different positions of
  the respective counters. Which means that $\Vv_{l+1}$ will receive a
  different number of $\S_l$ letters from $\rlp$ and $\brlp$ so it
  will get to state $\mathit{loop}$. The rest of the argument is exactly
  the same as in the previous case.
\end{proof}

We will show that in order to win in $\Cc^l$ the system  has no
other choice than to generate an iterated $l$-counter. Before this we
present a  general useful lemma:

\medskip

\begin{lemma}\label{l:sub}
  Consider a plant $\Cc$ consisting of two plants $\Cc_1$ and $\Cc_2$
  over process set $\PP_1$ and $\PP_2$, respectively. We assume that
  there exist $r_1 \in \PP_1$ and $r_2 \in \PP_2$ such that each
  action $a$ in $\Cc$ is such that either $\loc(a) \subseteq \PP_1$ or
  $\loc(a) \subseteq \PP_2$, or $\loc(a) \subseteq
  \set{r_1,r_2}$. Then every winning strategy in $\Cc$ gives a winning
  strategy in $\Cc_1$.
\end{lemma}

\medskip

\begin{proof}
  Just fix the behaviour of the environment in $\Cc_2$ and play the
  strategy in $\Cc$.
\end{proof}

With this at hand we can now prove the main lemma.

\medskip

\begin{lemma}\label{lemma:strategy-long}
  If $\s$ is a winning strategy in $\Cc^{l+1}$ and $x$ is a
  $\s$-play with no question then the projection of $x$ on
  $\bigcup_{i=1, \ldots,l+1}\S^\#_i$ is an iterated $(l+1)$-counter.
\end{lemma}

\medskip

\begin{proof}
  By the construction of $\Cc^{l+1}$, if
  there is no question during a $\s$-play, then the play is uniquely
  determined by the strategy. We will show that this unique play is an
  iterated $(l+1)$-counter. 

  By applying Lemma~\ref{l:sub} twice we obtain from $\s$
  a winning strategy in
  $\Cc^l$. By induction assumption the projection of $x$ on
  $\bigcup_{i=1, \ldots,l} \S^\#_i$  is an iterated
  $l$-counter. 
  Thus, between every two  consecutive
  $\$_l$ we have a letter from $\S_{l+1}$, followed by an $l$-counter and
  $\#_l$ (as long as we stay in the main part). The same
  holds for the $\bar{r_{l+1}}$ part. It remains to
  show that the sequence $u_0,u_1,\ldots$ of these $l$-counters 
  represents the values $0,1,\ldots$ modulo $\Tower(2,l)$, and the same
  for the sequence $\bar{u_0},\bar{u_1},\ldots$

  Assume that this is not the case and let $i$ be the index where the
  first error occurs. We will construct a play winning for the
  environment. 

  Let us first assume that the value of $\bar{u_i}$ is correct but the
  one of $u_i$ is not. Let $k$ be the first position where the error
  occurs in the $u_i$ counter. After the $k$-th letter of $u_i$ is
  transmitted to $r_{l+1}$ the environment can execute action
  $(\eqtest,c)$. Similarly, in process $\bar{r_{l+1}}$ after the
  $k$-th letter the environment can execute
  $\bar{(\eqtest,1-c)}$. Notice that these two questions are
  concurrent and happen after the letters of the corresponding
  counters are generated. Process $\Vv_{l+1}$ goes
  to $\mathit{neqtest}$ since it receives $(\eqtest,c)$, and
  $\bar{(\eqtest,1-c)}$. On the other levels the environment does not
  choose test actions. By induction, processes $r_l$ and $\bar{r_l}$
  will continue to generate iterated $l$-counters, since there are no
  questions in $\Cc^l$ and $\bar{\Cc^l}$.  As the environment
  has chosen the same position $k$ in both counters, process $\Vv_{l+1}$ will
  receive the same number of letters from $r_{l+1}$ and
  $\bar{r_{l+1}}$ thus entering into a rejecting state. This
  contradicts the assumption that the strategy in $\Cc^{l+1}$ was
  winning.

  The second case is where the value of $u_i$ equals $i
  \pmod{\Tower(2,l)}$, but the one of $\bar{u_{i+1}}$ is different from
  $(i+1) \pmod{\Tower(2,l)}$. Let $k$ be the position of the first
  error. In this case the environment can execute actions
  $\bar{(\inctest,c)}$, and $(\inctest,c)$ or $(\inctest,1-c)$, depending
  on whether or not there is some $a_l$ before position $k$ in
  $u_i$. As in the case above, these two questions are concurrent because
  process $\Vv_{l+1}$ first synchronizes with $\bar{r_{l+1}}$ and then with
  $r_{l+1}$. The same argument as above  shows that in
  this case we could find a play consistent with $\s$ and
  winning for the environment. 
\end{proof}

Putting Lemmas~\ref{lemma:strategy-exits}
and~\ref{lemma:strategy-long} together we obtain:

\medskip

\begin{proposition}\label{prop:game corrctness}
  For every $l$, the system has a winning strategy in $\Cc^l$. For
  every such winning strategy $\s$, if we consider the unique $\s$-play without
  questions then its projection on $\bigcup_{i=1,\ldots,l} \S^\#_i$
  is an iterated $l$-counter.
\end{proposition}

\medskip

\begin{theorem}\label{th:nonelem}
  Let $l>0$.  There is an acyclic architecture of diameter $2l+1$ and
  with $(2^{l+3}-3)$ processes such that the space complexity of the control
  problem for it is $\Omega(\Tower(n,l))$-complete. 
\end{theorem}

\medskip

\begin{proof}
  First observe that the plant $\Cc^l$ has $(2^{l+2}-3)$ processes
  and diameter $2l-1$. It is straightforward to make the
  $l$-counter count till $\Tower(n,l)$ and not to $\Tower(2,l)$ as we
  have done in the above construction. For this it is enough to make
  the $1$-counter count to $n$ instead of just to $2$.

  We will simulate space bounded Turing machines.  Take a machine $M$
  and a word $w$ of length $n$. We want to reduce the problem of
  deciding if $w$ is accepted by $M$ to the problem of deciding if the
  system  has a winning strategy for a plant $\Cc(M,w)$ of size
  polynomial in the sizes of $M$ and $w$. 

  A $\Tower(n,l)$ size configuration can be encoded by an
  $(l+1)$-counter. In an iterated $(l+1)$-counter we can encode a
  sequence of such configurations. The plant $\Cc(M,w)$ is obtained by a
  rather straightforward modification of the construction of
  $\Cc^{l+1}$\anca{change (l+2)  en (l+1)}. Instead of ensuring that the value of the first counter
  is $0$, it needs to ensure that it represents the initial
  configuration. Instead of ensuring that the two successive counters
  represent two successive numbers, it needs to ensure that they
  represent two successive configurations.  Using
  Proposition~\ref{prop:game corrctness}, the problem of deciding if
  a $\Tower(n,l)$-space bounded Turing machine $M$ accepts $w$ is
  polynomially reducible to the problem of deciding if the system has a
  winning strategy in the so obtained $\Cc(M,w)$. The size of $\Cc(M,w)$
  is exponential in $l$ and polynomial in $M,w,n$. The game can be
  constructed in the time proportional to its size.
\end{proof}

\section{Conclusions}

Distributed synthesis is a difficult and at the same time
promising problem, since distributed systems are intrinsically complex
to construct. We have considered  a simple, yet powerful model based on synchronization using shared memory --
as used in multithreaded programs or by  hardware primitives
such as compare-and-swap. 
Under some restrictions we have shown that the resulting control
problem is decidable. Since every process is allowed to interact with
the environment, our tree architectures are quite rich and allow to
model hierarchical situations, like server/clients. Such cases
are undecidable in the setting of Pnueli and Rosner.

%  We\todo{??}\ have used the setting of Ramadge and Wonham which 
% hid parts of the specification in the plant. 
% In our opinion Zielonka asynchronous automata are at least as interesting as the fully synchronous model of Pnueli and Rosner. 
% In our model we have insisted that control does not introduce new
% synchronizations: it does not reduce parallelism of the controlled
% system. It seems undesirable to have a solution that removes completely
% parallelism from the system. Even if one accepts to limit
% parallelism, it is not clear how to measure how much of it is left afterwards.

Already Pnueli and Rosner in~\cite{PR89icalp} strongly argue in
favour of asynchronous distributed synthesis.  
The choice of transmitting additional information while synchronizing
is a consequence of the model we have adopted. We think that it is
interesting from a practical point of view. It is also interesting
theoretically, since it allows to avoid simple (and unrealistic)
reasons for undecidability. Our lower bound result is somehow
surprising. Since we have full information sharing, all the
complexity must be hidden in the uncertainty about other processes
peforming in parallel. % The proof shows that even with three processes
% this uncertainty can be used to encode complex problems.  

Important problems remain open, in particular the decidability 
without the acyclic restriction. A more immediate task is to consider
non-blocking winning conditions and B\"uchi specifications. A further
interesting research venue is synthesis of open,
concurrent recursive programs, as considered e.g.~in~\cite{bgh09}.

\bibliographystyle{abbrv}
\bibliography{bibliography-new}

\begin{thebibliography}{10}

\bibitem{bgh09}
B.~Bollig, M.-L. Grindei, and P.~Habermehl.
\newblock Realizability of concurrent recursive programs.
\newblock In {\em FOSSACS}, volume 5504 of {\em LNCS}, pages 410--424, 2009.

\bibitem{church62}
A.~Church.
\newblock Logic, arithmetics, and automata.
\newblock In {\em Proceedings of the International Congress of Mathematicians},
  pages 23--35, 1962.

\bibitem{cgw12}
P.~Clairambault, J.~Gutierrez, and G.~Winskel.
\newblock The winning ways of concurrent games.
\newblock In {\em LICS}, pages 235--244. IEEE, 2012.

\bibitem{DieRoz95}
V.~Diekert and G.~Rozenberg, editors.
\newblock {\em The Book of Traces}.
\newblock World Scientific, 1995.

\bibitem{FinSch05}
B.~Finkbeiner and S.~Schewe.
\newblock Uniform distributed synthesis.
\newblock In {\em LICS}, pages 321--330. IEEE, 2005.

\bibitem{GLZ04}
P.~Gastin, B.~Lerman, and M.~Zeitoun.
\newblock Distributed games with causal memory are decidable for
  series-parallel systems.
\newblock In {\em FSTTCS}, volume 3328 of {\em LNCS}, pages 275--286, 2004.

\bibitem{GasSznZei09}
P.~Gastin, N.~Sznajder, and M.~Zeitoun.
\newblock Distributed synthesis for well-connected architectures.
\newblock {\em Formal Methods in System Design}, 34(3):215--237, 2009.

\bibitem{ggmw10}
B.~Genest, H.~Gimbert, A.~Muscholl, and I.~Walukiewicz.
\newblock Optimal {Z}ielonka-type construction of deterministic asynchronous
  automata.
\newblock In {\em ICALP}, volume 6199 of {\em LNCS}, 2010.

\bibitem{KatSchPel11}
G.~Katz, D.~Peled, and S.~Schewe.
\newblock Synthesis of distributed control through knowledge accumulation.
\newblock In {\em CAV}, volume 6806 of {\em LNCS}, pages 510--525. 2011.

\bibitem{kel73}
R.~M. Keller.
\newblock Parallel program schemata and maximal parallelism~{I}. {F}undamental
  results.
\newblock {\em Journal of the Association of Computing Machinery},
  20(3):514--537, 1973.

\bibitem{kv01}
O.~Kupferman and M.~Vardi.
\newblock Synthesizing distributed systems.
\newblock In {\em LICS}, 2001.

\bibitem{MadThiag01}
P.~Madhusudan and P.~Thiagarajan.
\newblock Distributed control and synthesis for local specifications.
\newblock In {\em ICALP}, volume 2076 of {\em LNCS}, pages 396--407, 2001.

\bibitem{MTY05}
P.~Madhusudan, P.~S. Thiagarajan, and S.~Yang.
\newblock The {MSO} theory of connectedly communicating processes.
\newblock In {\em FSTTCS}, volume 3821 of {\em LNCS}, 2005.

\bibitem{maz77}
A.~Mazurkiewicz.
\newblock Concurrent program schemes and their interpretations.
\newblock {DAIMI Rep. PB}~78, Aarhus University, Aarhus, 1977.

\bibitem{mel06}
P.-A. Melli\`es.
\newblock Asynchronous games 2: {T}he true concurrency of innocence.
\newblock {\em TCS}, 358(2-3):200--228, 2006.

\bibitem{MeyWil05}
R.~V.~D. Meyden and T.~Wilke.
\newblock Synthesis of distributed systems from knowledge-based specifications.
\newblock In {\em CONCUR, volume 3653 of LNCS}, pages 562--576, 2005.

\bibitem{PR89icalp}
A.~Pnueli and R.~Rosner.
\newblock On the synthesis of an asynchronous reactive module.
\newblock In {\em ICALP}, volume 372, pages 652--671, 1989.

\bibitem{PR90}
A.~Pnueli and R.~Rosner.
\newblock Distributed reactive systems are hard to synthesize.
\newblock In {\em FOCS}, pages 746--757, 1990.

\bibitem{RW89}
P.~J.~G. Ramadge and W.~M. Wonham.
\newblock The control of discrete event systems.
\newblock {\em Proceedings of the IEEE}, 77(2):81--98, 1989.

\bibitem{sf06}
S.~Schewe and B.~Finkbeiner.
\newblock Synthesis of asynchronous systems.
\newblock In {\em LOPSTR}, number 4407 in LNCS, pages 127--142. 2006.

\bibitem{SEM03}
A.~Stefanescu, J.~Esparza, and A.~Muscholl.
\newblock Synthesis of distributed algorithms using asynchronous automata.
\newblock In {\em CONCUR}, number 2761 in LNCS, pages 27--41, 2003.

\bibitem{zie87}
W.~Zielonka.
\newblock Notes on finite asynchronous automata.
\newblock {\em RAIRO--Theoretical Informatics and Applications}, 21:99--135,
  1987.

\end{thebibliography}

\end{document}